\documentclass[a4paper,11pt]{article}
\usepackage[utf8]{inputenc}
\usepackage[T1]{fontenc}
\usepackage[UKenglish]{isodate}

\usepackage{mathrsfs}
\usepackage{bm}
\usepackage{bbm}
\usepackage{palatino}
\usepackage[margin=2cm]{geometry}

\usepackage[usenames,dvipsnames]{xcolor}
\usepackage[pdftex]{graphicx}
\usepackage{amsmath,amssymb,amsfonts}
\usepackage{amsthm}
\usepackage[textsize=small,textwidth=1.3in]{todonotes}
\usepackage[english]{babel}

\usepackage[colorlinks=true,urlcolor=Mahogany,linkcolor=Mahogany,citecolor=Mahogany,plainpages=false,pdfpagelabels]{hyperref}

\usepackage{authblk}

\usepackage{multirow}

\usepackage{tikz}
\usepackage{verbatim}
\usetikzlibrary{shapes.geometric,plotmarks,backgrounds,fit,calc,circuits.ee.IEC}
\usetikzlibrary{decorations.pathreplacing}
\usepackage[backend=bibtex,maxcitenames=4,maxalphanames=100,maxbibnames=100,isbn=false]{biblatex}
\tikzset{phase/.style = {draw,fill,shape=circle,minimum size=5pt,inner sep=0pt},crossx/.style={path picture={ 
\draw[thick,black,inner sep=0pt]
(path picture bounding box.south east) -- (path picture bounding box.north west) (path picture bounding box.south west) -- (path picture bounding box.north east);
}}, cross/.style={path picture={ 
\draw[thick,black](path picture bounding box.north) -- (path picture bounding box.south) (path picture bounding box.west) -- (path picture bounding box.east);
}}, not/.style={draw,circle,cross,minimum width=0.3 cm}}
\usepackage{xcolor}

\usepackage{color, colortbl}
\usepackage{multirow}
\usepackage{easybmat}
\usepackage{physics}
\usepackage{epstopdf}

\usepackage{cleveref} %For multiple references using \cref

\usepackage{biblatex}
\usepackage{caption} %For the multicolumn figures
\usepackage{subcaption}

\newtheorem{theorem}{Theorem}
\newtheorem{lemma}[theorem]{Lemma}

\newtheorem{proposition}[theorem]{Proposition}
\newtheorem{definition}[theorem]{Definition}
\newtheorem{remark}[theorem]{Remark}

\providecommand{\keywords}[1]{\textbf{\textit{Keywords:}} #1}
\long\def\symbolfootnote[#1]#2{\begingroup\def\thefootnote{\fnsymbol{footnote}}\footnote[#1]{#2}\endgroup}

\begin{document}

\title{\vspace{-1.3cm}Polar Codes for Quantum Reading}
\author[1,2]{Francisco Revson F. Pereira\footnote{revson.ee@gmail.com}}
\author[1,2]{Stefano Mancini}
\affil[1]{School of Science and Technology, University of Camerino, I-62032 Camerino, Italy}
\affil[2]{INFN, Sezione di Perugia, I-06123 Perugia, Italy\vspace{-1cm}}

\date{}

%\institute{}
% \institute{Institutions}
%
\maketitle              % typeset the header of the contribution

\begin{abstract}
% The abstract should briefly summarize the contents of the paper in
% 150--250 words.
% 18-19 lines

Quantum reading provides a general framework where to formulate the statistical discrimination of quantum channels. 
Several paths have been taken for such a problem. However, there is much to be done
in the avenue of optimizing channel discrimination using classical codes.
At least two open questions can be pointed to: how to construct low complexity encoding schemes that are interesting for channel 
discrimination and, more importantly, how to develop capacity-achieving protocols. 
The aim of this paper is to present a solution to these questions using polar codes. 
Firstly, we characterize the rate and reliability of the channels under polar encoding. 
We also show that the error probability of the scheme proposed decays exponentially with respect to the code length. 
Lastly, an analysis of the optimal quantum states to be used as probes is given.

\newline
%\noindent {\bf MSC: }%polar_encodingquantum_polar_encoding81P70, 81P40, 94B15, 94B27.
\end{abstract}

\section{Introduction}
\label{sec:Introduction}
\noindent

Quantum hypothesis testing aims to identify strategies in order to statistically discriminate quantum states or processes. The former is called quantum state 
discrimination and has been largely analyzed in the literature (starting from Refs. \cite{Helstrom1969,Holevo1973}). The latter, under the name of quantum 
channel discrimination, has been recently addressed
\cite{Childs2000,Acin2001,Gilchrist2005,Sacchi2005,Sacchi2005a,Wang2006,Duan2009,Hayashi2009,Harrow2010}.
In its basic formulation, one has to identify a quantum channel selected from a set accordingly to a probability measure. This should be done by using 
suitable input state and output measurement. As such, it is a double-optimization problem and hence results in a daunting task. Performance are usually 
quantified in terms of minimum error probability and recently bounds on it were found for general strategies 
\cite{Pirandola2019,Katariya2020}. Although theoretically subtle, quantum channel discrimination is interesting for practical applications. 
For instance, it is at the basis of the decoding procedure of two-way quantum cryptography \cite{Pirandola2008} where the secret 
information is encoded in a Gaussian ensemble of phase-space displacements. It also appears in the quantum illumination of targets 
\cite{Lloyd2008,Tan2008}, where the sensing of a remote low-reflective object in a bright thermal environment corresponds to the binary 
discrimination between a lossy channel (presence of target) and a depolarizing channel (absence of target). Following this line, quantum 
channel discrimination can be reformulated in the framework of quantum reading~\cite{Pirandola2011}.
% (although quantum rebound is also used~\cite{Das2019}).
There, the data storage corresponds to a process of channel encoding, where information is recorded into a cell by storing a quantum channel 
picked from a given ensemble. Then readout corresponds to the process of channel decoding, which is equivalent to discriminating 
between the various channels in the ensemble. In such a setting, using a quantum resource, such as entanglement, was shown 
to overcome any classical strategy based on mixtures of coherent states, hence the name of ``quantum'' reading. 

Efficient paths to quantum reading can be also envisaged by the use of coding techniques \cite{Pirandola2011a,Das2019}. There are at least
two possible approaches: classical coding on the quantum memory cell labels and quantum coding 
on the probe states used in the readout. This paper takes the former one. Similar to channel coding, where redundancy is added to protect the information 
to be transmitted, we use classical codes to encode the information to be recorded in the memory cell. The goal is to decrease the error probability in discriminating 
the channels in the ensemble. 
There are several classical codes that one could use for quantum reading, depending on the goal in mind. In this work, we are interested in 
information-theoretically provable codes that could be used to attain high information rates and low error probability. Therefore, suitable codes can be 
derived from the family of polar codes.

Creating capacity-achieving codes has been a challenge since the development of Classical Information Theory. Even more difficult 
was to show there exist theoretically-proven such codes. Fortunately, developing the idea of information combining, Arikan was able to 
show that for binary memoryless symmetric (BMS) channels such codes exist~\cite{arikan09}. They are called polar codes. 
The major achievement of Arikan's paper is to show a clever way to combine information, which leads to new 
synthetic channels manifesting a polarization phenomenon.
They are commonly divided into two groups named ``good'' and ``bad'' channel\footnote{The reason for such names is due to the capacity of 
these new channels been close to its maximum value and close to zero, respectively.}. Furthermore, the fractions of good and bad channels are related to the 
capacity of the original channel in consideration. Encoding and decoding schemes take advantage of these properties to attain low complexity. 
It is worth mention that polar code construction highly depends on the channel into consideration.

Several papers extend or apply the idea of polarization in a diversity of channels and areas~\cite{sasoglu09,arikan10,sahebi11,honda13}.
In particular, polar codes have also been constructed to classical-quantum channels~\cite{Wilde2013}, and quantum channels~\cite{Wilde2013a,Renes2015}. 
However, this is not the case when one considers quantum reading, which is the goal of this paper.
The distribution of the channels to be discriminated can be symmetric or asymmetric. Asymmetric quantum reading can be seen as a generalization of the 
former quantum reading and this formulation can be justified by the measurement strategy implemented in the decoding process or energy constrain 
imposed over the channels. Since these two conditions are plausible hypothesis to be considered in classical digital memories, which are the main goal of channel model used in this paper, we will formulate our results for asymmetric quantum reading and, where it makes necessary, an adaptation for symmetric quantum reading is given. 
Therefore, our polar code will be design to, primarily, asymmetric quantum channels. 
In the classical paradigm, this has been done by Sutter, et al. \cite{Sutter2012}, and Honda and Yamamoto \cite{honda13}. Some of the ideas used in Ref.~\cite{honda13} are applied in this paper. However, as we will see, the transition from the classical to quantum paradigm is subtle and major changes are needed to be done to derive our construction. 
This is particularly true for the measurement process implemented in the decoding part. 
An adequately definition for rate and reliability is also needed. 

The analysis of quantum memory cell in this paper is twofold. Initially, polar coding is applied on the labels of the quantum memory cell in order to decrease 
the error probability in distinguishing them. In this direction, we show that the polarization phenomenon can be characterized in quantum memory cell by 
studying its composing parts; namely, channel combining and channel splitting. The former constitutes a systematic
approach to combine source bits so the polarization phenomenon emerges. An encoding matrix, sometimes also called combining function, describing the process is given in this paper.
The latter is an information-theoretically analysis of the synthesized channels 
created in the channel combining part. Initially, the first level encoding process combining two quantum channels is analyzed. We show that the rate and reliability 
of the synthesized channels polarize. This is later used in the asymptotic analysis. A connection between rate and reliability is also presented, showing that rate is inversely 
proportional to reliability. Next, our first major result is given. It characterizes the asymptotic distribution of the synthesized channel with respect to the length code. 
This result is based on proving the existence, shown in this paper, of a symmetric quantum channel that can be used to obtain the statistics of the (asymmetric) quantum channel of interest.
Lastly, in this first part, we draw the encoding and decoding scheme. The strategy used for encoding the frozen bits does not follow the scheme used in previous works on
polar codes for classical-quantum channels. Additionally, examining the error probability obtained after the decoding scheme, 
we see that it decays exponentially with respect to the length code. The second part addresses the optimization problem of probe states. As a first-order approximation of classical digital memories, 
we consider amplitude damping channels as our channel model. The fundamental result of this part 
is showing that the optimal probe states are pure states. Even though we have taken into consideration just the first level polar encoding, it seems satisfactory supposing this 
can be extended to any $N$-level polar encoding.

This paper is organized as follows. In Section~\ref{sec:Preliminaries} we present previous results on polar codes for classical-quantum channels, 
and some important definition used to characterize quantum memory cells. This section constitutes the fundamentals that our work in based on. 
Next, the main results are shown in Section~\ref{sec:cCodingScheme}. Encoding and decoding schemes are described and analyzed in detail, showing 
that the use of polar coding is also interesting for quantum memory cells. The following subject treated in this work is the optimization of 
states used to probe the channels. This is drawn in Section~\ref{sec:inputDependenceAnalysis}. 
Lastly, we draw our conclusions and some final remarks in Section~\ref{sec:Conclusion}.

\subsection{Notation}

We denote classical random variables as $X$, $Y$, $U$, whose realizations are elements of the 
finite sets $\mathcal{X}$, $\mathcal{Y}$, $\mathcal{U}$, respectively. The probability distributions are respectively represented by 
$p_X(x)$, $p_Y(y)$, and $p_U(u)$ for the random variables $X$, $Y$, $U$. In particular, $X$ is always assumed to be discrete and, in some parts of the text, to be a Bernoulli 
random variable with values in $\{0, 1\}$ and $P_X(0) = p$ by using the notation $X\sim \text{Ber}(p)$. 
For such random variables, its binary Shannon entropy is defined as $H(X) = h(p):=-p\log(p)-(1-p)\log(1-p)$.
The use of subscript and superscript on a letter indicates a sequence starting with the element denoted by the subscript and ending with the element denoted by the superscript, 
e.g. $X_{i}^{j}$ (with $j> i$) is the sequence of random variables $X_{i},X_{i+1},\ldots,X_{j}$. $X^{N}$ simply stands for the sequence $X_{1},X_{2},\ldots,X_{N}$.
A classical memoryless channel is written as 
$W\colon\mathcal{X}\rightarrow\mathcal{Y}$ from the input alphabet $\mathcal{X}$ to the output alphabet $\mathcal{Y}$. We write $W^N$ to 
denote $N$ uses of the channel $W$; i.e., $W^N\colon\mathcal{X}^N\rightarrow\mathcal{Y}^N$ with $W^N(y^N|x^N) = \prod_{i=1}^N W(y_i|x_i)$ 
since $W$ is memoryless.
Quantum systems $A$, $B$, and $C$ correspond to Hilbert spaces 
$\mathcal{H}_A$, $\mathcal{H}_B$, and $\mathcal{H}_C$. The notation $A^N:= A_1A_2\cdots A_N$ denotes a joint system consisting of $N$ subsystems, each of which is isomorphic to $\mathcal{H}_A$. Let $\mathcal{L}(\mathcal{H}_A)$ denote the algebra of bounded linear operators acting on a 
Hilbert space $\mathcal{H}_A$. The subset $\mathcal{L}_+(\mathcal{H}_A)$ of $\mathcal{L}(\mathcal{H}_A)$ denotes the set of all positive semi-definite operators. A special and important class of operators in $\mathcal{L}_+(\mathcal{H}_A)$ is the one containing density operators 
$\mathcal{D}(\mathcal{H}_A)$. A density operator $\rho_A\in\mathcal{D}(\mathcal{H}_A)$ is a positive semi-definite operator with unit trace, 
$\text{Tr}\{\rho_A\} = 1$, and represents the state of a quantum system $A$. Lastly, a quantum channel $\mathcal{W}$ is a linear completely positive trace-preserving map from 
$\mathcal{L}(\mathcal{H}_A)$ to $\mathcal{L}(\mathcal{H}_B)$.

%The End 

\section{Preliminaries}
\label{sec:Preliminaries}
This section is devoted to introducing the main concepts of polar codes and quantum reading. 
We begin with a brief overview of binary polar codes and their desirable attributes. They can be used to 
achieve the capacity of discrete memoryless channels (DMC); there are encoding and decoding efficient schemes; and it is possible to reach 
error probabilities that decay exponentially in the square root of the blocklength. 
In the following, quantum reading problematic is introduced. The general concept is given, followed by the channel model 
adopted for analyzing probe states. In the end, we present the quantities assumed as rate and reliability.

%%%%%%%%%%%%%%%%%%%%%%%%%%%%%%%%%%%%%%%
%----------- First Section -----------%
%%%%%%%%%%%%%%%%%%%%%%%%%%%%%%%%%%%%%%%
\subsection{Classical-Quantum Polar Codes}
\label{sec:polarCodes}

Let $W\colon \mathcal{X}\rightarrow\mathcal{D}(\mathcal{H}_B)$ be a binary-input memoryless classical-quantum channel (cq channel) with input and 
output alphabets given by $\mathcal{X}$ and $\mathcal{D}(\mathcal{H}_B)$, respectively. Suppose the random variable $X$ is $\text{Ber}(\frac{1}{2})$. 
Then the transition probabilities can be derived from the joint input-output state 
\begin{equation}
\rho^{XB} = \frac{1}{2}\ketbra{0}{0}^X\otimes\rho_0^B + \frac{1}{2}\ketbra{1}{1}^X\otimes\rho_1^B,
\end{equation}where $\rho_{x}^B = W(x)$, for $x=0,1$. Notice that we are using superscripts to emphasize the system to which each part belongs. 
Whenever this is clear from context, they will be omitted. One important characterization can be given to this cq channel:
since the possible outcomes are uniformly distributed over all possible labels $x\in\mathcal{X}$, 
we say that the cq channel $W$ is symmetric. Similarly, a cq channel is called asymmetric when this does not happen. Though in this subsection we are going 
to deal with a symmetric channel, the rest of the text adopts asymmetric channel in order to provide a more general formulation. As a last comment to be made, 
we adopt as input alphabet $\mathcal{X} = \mathbb{Z}_2$, and arbitrary finite-dimensional output 
alphabet $\mathcal{D}(\mathcal{H}_B)$ and transition probability. This choice of input alphabet allows us to operate with their elements; in particular, 
we can use XOR operations, or sum$\mod 2$ operations, with the elements of the input alphabet. An appropriate combination of these operations 
leads to channel polarization.

Channel polarization consists mainly of two parts. The first one is named channel combining, which describes a method of combining inputs of 
$N$ cq channels. The second is channel splitting. This part is an information-theoretical analysis of new inputs and outputs that the channel combining 
produces. These new inputs and outputs generate synthesized channels. With a careful examination the synthesized channels, it is possible to show 
that, for an arbitrarily large number of them, they fit into two sets called good and bad channels. The statistical behavior of them gives the 
desirable attributes of polar codes~\cite{arikan09,arikantelatar09,honda13,Wilde2013,Wilde2013a,Renes2015}: they achieve the capacity when used for transmitting information 
over a cq channel; they can be encoded efficiently (with a complexity that is essentially linear in the blocklength); the error probability of 
the decoder decays exponentially in the square root of the blocklength. A descriptive explanation of channel polarization, synthesized channels, 
and some attributes of polar codes are given below.

Suppose there are $N$ copies of a cq channel $W$, which we denote by $W^N$, and $N$ realizations $u^N$ of a random variable $U$
representing the source. 
The general formulation of polar codes consists of applying a composition function to the input $u^N$, traditionally represented by a matrix 
$G_N$, and using the output of the composition function, defined as $x^N$, as the actual inputs for the channels $W^N$. This general procedure is 
called channel combining, and a scheme is shown in Fig.~\ref{fig:general_polar_coding}. Observe that this is the same scheme used in any linear channel coding; 
the type of the composition function is what determines the coding scheme to generate a polar code. 
For a particular example of coding scheme and composition function, with $N = 2$, see Fig.~\ref{fig:polar_coding}.

\begin{figure}[h!]
\begin{center}
\includegraphics[width=0.4\linewidth]{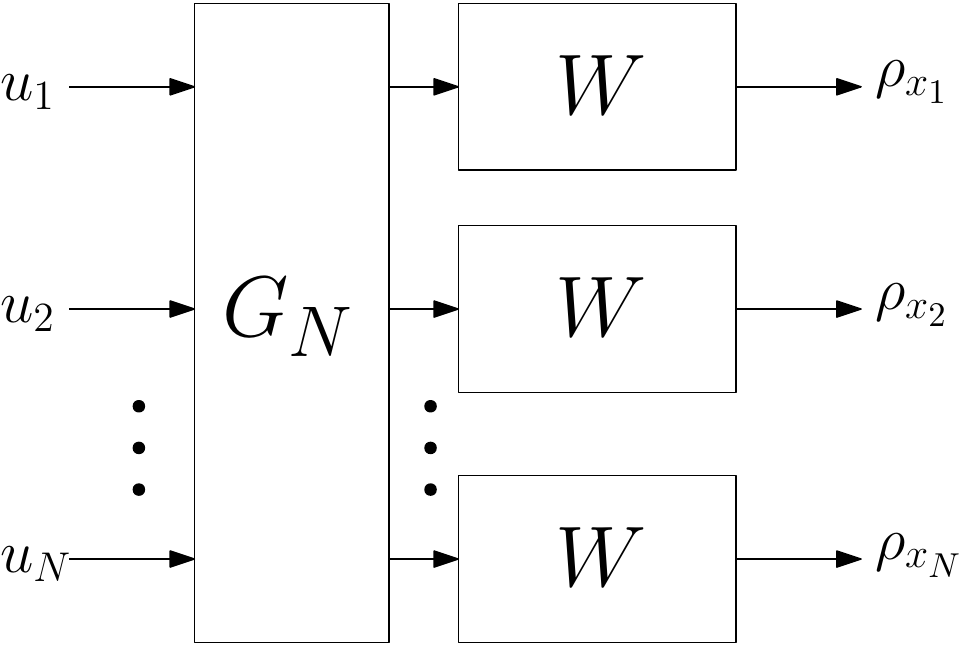}
\end{center}
\caption{General polar encoding scheme for a classical-quantum channel. The matrix $G_N$ represents the composition function applied on the input $u^N$ 
and resulting in the vector $x^N = u^N G_N$. The density operator $\rho_{x^N} = \rho_{x_1}\otimes\cdots\otimes\rho_{x_N}$ is the output $W^N(x^N)$.}%
\label{fig:general_polar_coding}%
\end{figure}

\begin{figure}[h!]
\begin{center}
\includegraphics[width=0.4\linewidth]{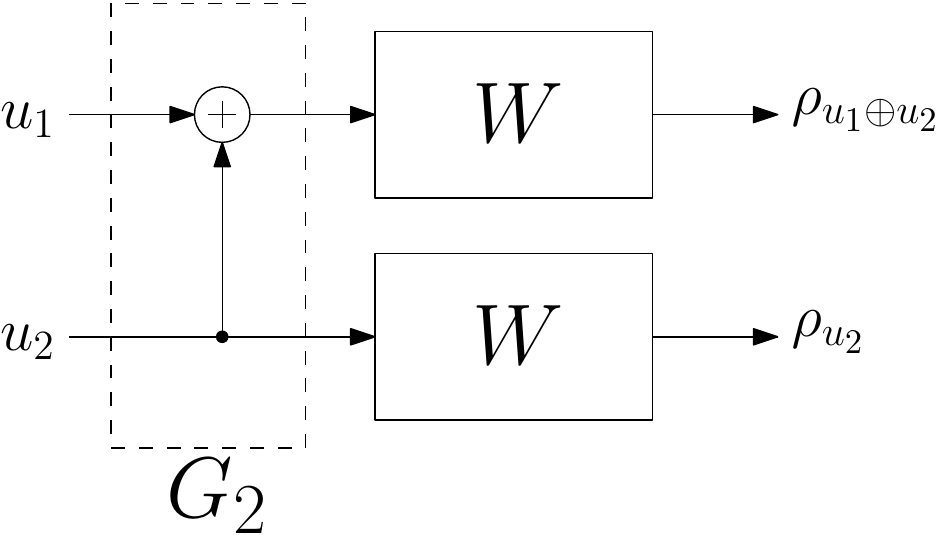}
\end{center}
\caption{Polar encoding scheme for $N = 2$. The choice of the composition function characterizes the encoding scheme to be polar.}%
\label{fig:polar_coding}%
\end{figure}

Now, we can introduce channel splitting and the polarization phenomenon that emerges from it. For $i\in\{1,\ldots,N\}$, 
we define the $i$-th synthesized channel $W_N^{(i)}$ with input alphabet $\mathcal{U}$ and output $\mathcal{D}(\mathcal{H}_{U_1^{i-1}}\mathcal{H}_{B^N})$ as

\begin{eqnarray}
W_N^{(i)}(u_i) &=& \sum_{u_1^{i-1}}\frac{1}{2^{i-1}}\ketbra{u_1^{i-1}}{u_1^{i-1}}^{U_1^{i-1}}\otimes\overline{\rho}_{u_1^{i-1}}^{B^N},\\
\overline{\rho}_{u_1^{i-1}}^{B^N} &=& \sum_{u_{i+1}^N}\frac{1}{2^{N-i}}\rho_{u^N G_N}^{B^N},
\end{eqnarray}where $N$ is a power of two and $G_N$ is the composition function. This formulation comes from analyzing a successive decoder acting on the output channel 
with the help of a genie. For more details, have a look at the original formulation by Arikan~\cite{arikan09} or its extension to classical-quantum channels in Ref.~\cite{Wilde2013}.

There are two important quantities used to quantify the polarization phenomenon and error probability decay in polar coding. These quantities are called rate and reliability. 
Their definition depends on the channel under consideration. For a classical-quantum channel $W$, rate is defined as the mutual information $I(X;B)$ and we denote it as $I(W)$. 
Reliability is adopted as the fidelity between the possible channel outputs, i.e., reliability is given by $F(\rho_0,\rho_1) = ||\sqrt{\rho_0}\sqrt{\rho_1}||_1$, where $||A||_1$
is the Schatten 1-norm of a operator $A\in\mathcal{L}(\mathcal{H}^B)$. It is possible to show that these quantities are inversely proportional to each other in the sense that when 
one has a value close to its maximum, the other has a value close to zero~\cite{Wilde2013}. Some interesting results are obtained studying the rate and reliability over channels 
produced in the channel splitting part. In particular, for $N$ sufficiently large, it is shown that the channels $W_N^{(i)}$ are divided in two sets: 
one set with $I(W_{N}^{(i)})$ close to unit, these channels are called ``good''; and a set with $I(W_{N}^{(i)})$ close to zero, called ``bad'' channels. 
We denote by $\mathcal{A}$ the set with indexes labeling good channels.
Analyzing how the fraction of good and bad channels grow when $N$ goes to infinity, it is shown the polarization phenomenon and a capacity-achieving coding strategy for polar codes. 
See the following proposition for a formal formulation of the exposed ideas.

\begin{proposition}\cite[Thm. 2, Prop. 4]{Wilde2013}
Let $W$ be a classical-quantum channel. Then the following is true:
\begin{enumerate}
	\item The channels $\{W_N^{(i)}\}$ polarize in the sense that, for any $\delta\in(0,1)$, as $N$ goes to infinity through powers 
	of two, the fraction of indexes $i\in\{1,2,\ldots,N\}$ for which $I(W_{N}^{(i)})\in(1-\delta,1]$ goes to $I(W)$ and the fraction for 
	which $I(W_N^{(i)})\in[0,\delta)$ goes to $1-I(W)$;
	\item For any choice of parameters $(N,K,\mathcal{A})$ for a classical-quantum polar code, the probability of error is bounded above by
	\begin{equation}
	P_e(N,K,\mathcal{A})\leq 2\sum_{i\in\mathcal{A}}\sqrt{F(W^{(i),0},W^{(i),1})}.
	\end{equation}In particular, for any fixed $R = K/N< I(W)$ and $\beta<1/2$, block error probability for polar coding under sequential  
	decoding satisfies
	\begin{equation}
	P_e(N,K)= o(2^{-N^\beta}),
	\end{equation}where $o(\cdot)$ is the little-O notation from complexity theory.
\end{enumerate}
\label{Proposition:Preliminaries}
\end{proposition}

We make use of Proposition~\ref{Proposition:Preliminaries} in the following section to show that polar codes can be applied 
in the task of quantum memory cell discrimination.

%%%%%%%%%%%%%%%%%%%%%%%%%%%%%%%%%%%%%%
%---------- Second Section ----------%
%%%%%%%%%%%%%%%%%%%%%%%%%%%%%%%%%%%%%%
\subsection{Quantum Memory Cell}\
\label{sec:quantum_memory}

We now formulate the general description of quantum memory cell used in this paper. A \textit{quantum memory cell} is defined as the set 
$\{\mathcal{W}^x\}_{x\in\mathcal{X}}$ of quantum channels. For a fixed $x$, we have

\begin{eqnarray}
\mathcal{W}^x\colon \mathcal{D}(\mathcal{H}_{B'})&\rightarrow \mathcal{D}(\mathcal{H}_B)\\
\rho&\mapsto \mathcal{W}^x(\rho),
\end{eqnarray}where $\mathcal{D}(\mathcal{H}_{B'}),\mathcal{D}(\mathcal{H}_B)$ are the sets of input and output 
density states of the quantum channel $\mathcal{W}^x$. We call $x\in\mathcal{X}$ the \textit{quantum memory cell index}. 
Sometimes we denote $\mathcal{W}^x$ as $\mathcal{W}_{{B'}\rightarrow B}^x$ to highlight the input and output systems of $\mathcal{W}^x$.
An important hypothesis is given in here. We are supposing that the distribution of the random variable $X$ describing the label of the quantum channels is 
non-uniform. This can be justified by the measurement strategy implemented in the decoding process or energy constrain imposing this distribution. Since 
these two reasoning are plausible in this paper, we are going to adopt, in most parts of this paper, that we are dealing with asymmetric quantum reading. 
The definition of asymmetric quantum reading goes similarly to the definition of asymmetric cq channels. 
A quantum memory cell $\{\mathcal{W}^x\}_{x\in\mathcal{X}}$ is called asymmetric if the possible outcomes are non-uniformly distributed over all possible labels $x\in\mathcal{X}$.
Lastly, notice that minor changes over our results lead to the applicability of them to symmetric quantum reading. 

\begin{remark}
Before introducing the channel model used for numerical analysis in this paper, an important point-of-view over quantum memory cell need to be given. A common approach to 
quantum memory cell is to study how to discriminate their elements using different probe states but having the labels fixed. However, in this paper, we are also interested in improving  
quantum memory cell discrimination by working on the channel labels. This optimization is performed by fixing the probe states used. Along these lines, quantum memory cell can be seen 
as a set of classical-quantum channels. Thus, a capacity-achieving approach using polar codes is possible when one deals with quantum memory cell elements having indexes obeying some rule. 
The rule used in this paper is given by the composition function of polar codes.
\end{remark}

For this paper, we will adopt the \textit{amplitude damping} (AD) channel as our channel model. This is because AD channel is the finite-dimensional 
first-order approximation of bosonic attenuator channel, which is the standard model for digital memory cell. They can be described in the following 
manner. Let $\rho\in \mathcal{D}(\mathcal{H})$ be a single-qubit density state, then the AD channel can be described by the following Kraus expression:
\begin{equation}
\mathcal{W}^x(\rho) = A_0\rho A_0^\dagger + A_1\rho A_1^\dagger,
\label{Eq:ADchannel}
\end{equation}where 

\begin{equation}
A_0 = 
\begin{pmatrix}
1 & 0\\
0 & \sqrt{1-x}
\end{pmatrix},\quad
A_1 = 
\begin{pmatrix}
0 & \sqrt{x}\\
0 & 0
\end{pmatrix}.
\end{equation}Observe that the same notation for general quantum memory cell and AD channel has been used. However, through the paper will be 
clear when we are talking about one or the other.

The goal of the following sections is to show that using polar codes it is possible to attain optimal rate with low reliability. 
This is possible by using synthesized channels formulation and connecting its properties with the original channel under analysis. 
A method for approaching this is firstly introducing the joint input-output density state and characterizing the channel via this 
density state. Thus, lets $X$ be a random variable with probability law $p_X$. 
We can write the joint density state describing the systems $X$ and $B$ as

\begin{equation}
\rho^{XB} = p_X(0)\ketbra{0}{0}^X\otimes\mathcal{W}^0(\rho) + p_X(1)\ketbra{1}{1}^X\otimes\mathcal{W}^1(\rho).
\end{equation}As outlined previously, there are two parameters at the center of polarization phenomenon: rate and reliability. Rate is defined in this work 
as the quantum mutual information between the source $X$ and the output system $B$:

\begin{definition}
\label{def:rate}
Let $X\sim \text{Ber}(p)$ and $\mathcal{W}_{B'\rightarrow B}^x(\rho)$ be a quantum memory cell, where $x\in\mathcal{X}$.
The \textit{rate} of $\mathcal{W}$ is defined as $I(\mathcal{W})_\rho := I(X;B)_\rho$. 
A direct computation of this quantum mutual information shows that 
\begin{equation}
I(\mathcal{W})_\rho = H\Big{(}p\mathcal{W}^0(\rho) + (1-p)\mathcal{W}^1(\rho)\Big{)} - pH(\mathcal{W}^0(\rho)) - (1-p)H(\mathcal{W}^0(\rho)),
\end{equation}where $H(\sigma)$ is the von Neumann entropy for a density operator $\sigma\in\mathcal{D}(\mathcal{H})$.
\end{definition}

For reliability of the quantum memory cell $\mathcal{W}$, it is used the fidelity between the possible output states:

\begin{definition}
\label{def:reliability}
Let $X\sim \text{Ber}(p)$ and $\mathcal{W}_{B'\rightarrow B}^x(\rho)$ be a quantum memory cell, where $x\in\mathcal{X}$.
The \textit{reliability} of the quantum memory cell $\mathcal{W}$ is defined as
\begin{equation}
Z(\mathcal{W})_\rho := 2\sqrt{p(1-p)}F(\mathcal{W}^0(\rho),\mathcal{W}^1(\rho)) = 2\sqrt{p(1-p)}||\sqrt{\mathcal{W}^0(\rho)}\sqrt{\mathcal{W}^1(\rho)}||_1,
\end{equation}where $F(\rho,\sigma) = ||\sqrt{\rho}\sqrt{\sigma}||_1$ is 
the fidelity, and $||A||_1$ is the Schatten 1-norm of a operator $A\in\mathcal{L}(\mathcal{H})$.
\end{definition}

The following section describes the polar coding scheme proposed in this paper and applies the rate and reliability defined above to quantify 
the goodness of the codes created.

\section{Polar Coding Scheme}
\label{sec:cCodingScheme}
In this section, we consider the polarization phenomenon induced by a combining function acting on the quantum memory cells indexes. This approach
is similar to the classical-quantum polarization explained in the previous section. However, significant refinements in the arguments and proofs are needed
for the results to be valid in our situation.

\subsection{Channel Polarization}
	\subsubsection{Channel Combining}
	\label{subsec:channel_combining}

As mentioned before, for a fixed quantum probe state, the elements $\mathcal{W}^x$, for $x\in\mathbb{Z}_2$, of a quantum memory cell can be seen as a classical-quantum channel.
Then, without loss of generality, we will treat them in this form through this section. For an illustrative description of
our coding scheme in conjunction with our hypothesis, see Fig.~\ref{fig:quantum_polar_encoding_scheme_general}. Let $\mathcal{W}$ be a
classical-quantum channel from which we derive an $N$-fold classical-quantum channel $\mathcal{W}_N$ recursively, where $N = 2^n$ with $n\in\mathbb{N}_0$.
The zeroth level of recursion gives solely the channel $\mathcal{W}_1(u) = \mathcal{W}^u(\rho)$, for all
$u\in\mathbb{Z}_2$. The first level is a composition of two zeroth level channels; i.e.,
the classical-quantum channel $\mathcal{W}_2$ is given by

\begin{eqnarray}
\mathcal{W}_2(u_1,u_2) &=& \mathcal{W}_1(u_1\oplus u_2)\otimes \mathcal{W}_1(u_2)\\
&=& \mathcal{W}^{u_1\oplus u_2}(\rho)\otimes \mathcal{W}^{u_2}(\rho).
\end{eqnarray}This is shown in Fig.~\ref{fig:classical-quantum_polar_encoding}.

\begin{figure}[h!]
\begin{center}
\includegraphics[width=0.5\linewidth]{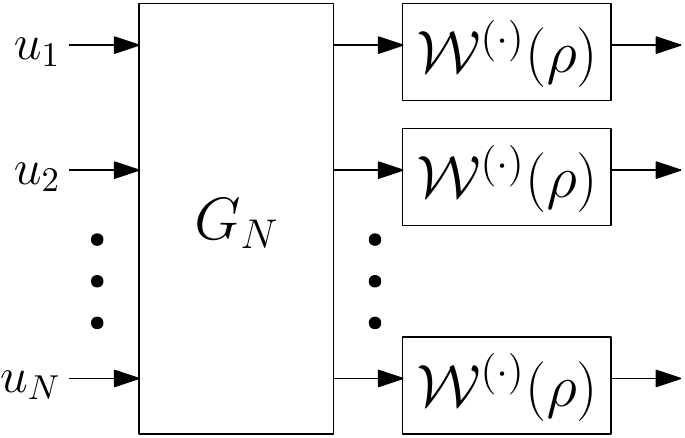}
\end{center}
\caption{Classical polar coding scheme.}%
\label{fig:quantum_polar_encoding_scheme_general}%
\end{figure}

\begin{figure}[h!]
\begin{center}
\includegraphics[width=0.5\linewidth]{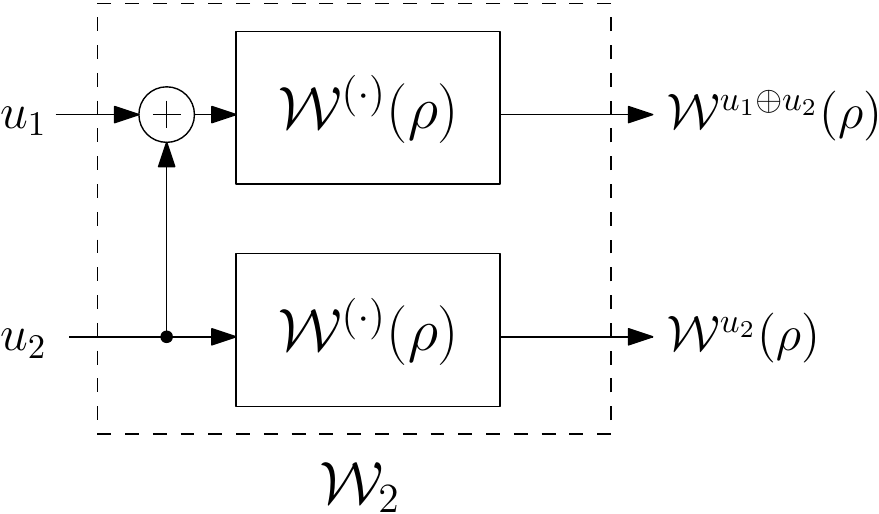}
\end{center}
\caption{Fundamental polar encoding block}%
\label{fig:classical-quantum_polar_encoding}%
\end{figure}

The second level follows from two copies of first level channels. The rule is
\begin{eqnarray}
\mathcal{W}_4(u_1,u_2,u_3,u_4) &=& \mathcal{W}_2(u_1\oplus u_2,u_3\oplus u_4)\otimes \mathcal{W}_2(u_2,u_4)\\
&=& \mathcal{W}^{u_1\oplus u_2\oplus u_3 \oplus u_4}(\rho)\otimes \mathcal{W}^{u_3\oplus u_4}(\rho) \otimes\mathcal{W}^{u_2\oplus u_4}(\rho)\otimes \mathcal{W}^{u_4}(\rho).
\end{eqnarray}We depict this scheme in Fig.~\ref{fig:classical-quantum_polar_encoding_4th-level}.

\begin{figure}[h!]
\begin{center}
\includegraphics[width=0.5\linewidth]{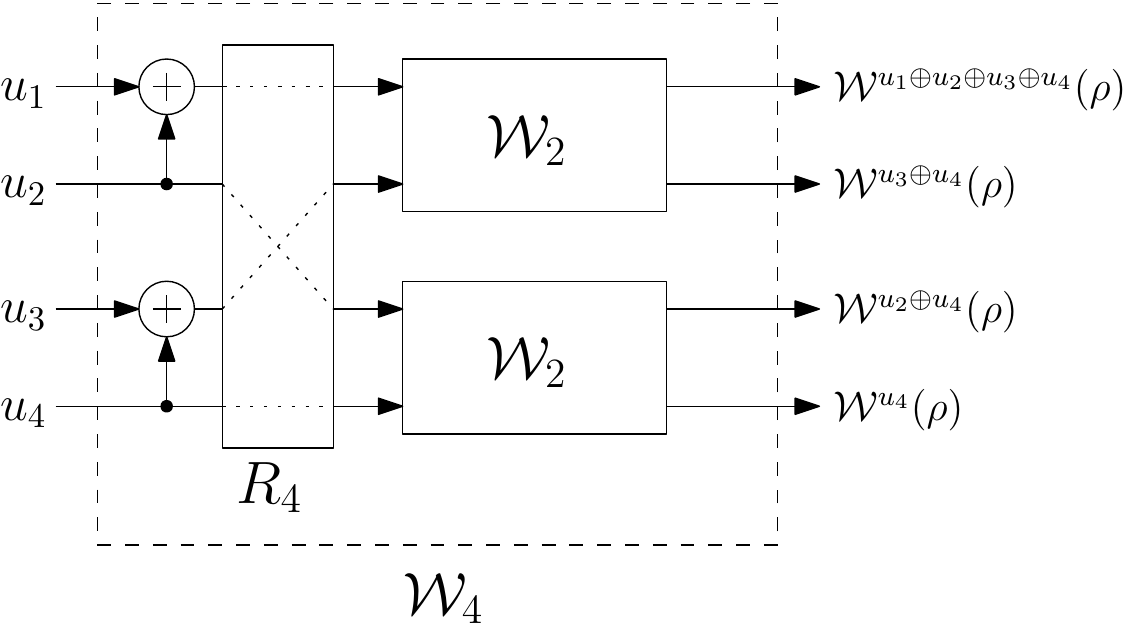}
\end{center}
\caption{Second level of polar coding scheme.}%
\label{fig:classical-quantum_polar_encoding_4th-level}%
\end{figure}

Following the same procedure, we can derive the $n$-level channels. As described in \cite{arikan09}, the matrix $G_N$, which connects
the source output $u^N$ to the channel input $x^N$ by $x^N = u^N G_N$, can be expressed as

\begin{equation}
G_N = R_N F^{\otimes n},
\end{equation}with $F^{\otimes n}$ being the $n$-fold Kronecker product of the matrix

\begin{equation}
F =
	\begin{pmatrix}
 		1 & 0 \\
 		1 & 1
	\end{pmatrix},
\end{equation}and $R_N$ is a permutation matrix known as \textit{bit-reversal} \cite{arikan09}. In particular, the second level combining function is

\begin{equation}
G_4 =
\begin{pmatrix}
 		1 & 0 & 0 & 0 \\
 		1 & 0 & 1 & 0 \\
 		1 & 1 & 0 & 0 \\
 		1 & 1 & 1 & 1
	\end{pmatrix}.
\end{equation}The next step is to characterize the synthesized channels produced by the action of $G_N$. This is the goal of the next subsection.

	\subsubsection{Channel Splitting}
	\label{subsec:channel_splitting}

First of all, consider two realizations of the channel $\mathcal{W}$. A new input-output relation can be generated through the channel combining procedure described before.
These channels are called synthesized since they are not real channels but new point-of-views obtained from the relations created.
We denote this transformation by $(\mathcal{W},\mathcal{W})\rightarrow(\mathcal{W}^-,\mathcal{W}^+)$,
where $\mathcal{W}^-$ and $\mathcal{W}^+$ are the synthesized channels. Fixing each input and examining the corresponding outputs, we come to the definition of
the synthesized channels below
\begin{subequations}
\label{Eq:synthesizedChannels}
\begin{eqnarray}
\mathcal{W}^-&\colon& u_1\in\mathbb{Z}_2\mapsto\mathcal{W}^{-,u_1}(\rho)\in\mathcal{D}(\mathcal{H}_{B_1}\otimes\mathcal{H}_{B_2}),\\
\mathcal{W}^+&\colon& u_2\in\mathbb{Z}_2\mapsto\mathcal{W}^{+,u_2}(\rho)\in\mathcal{D}(\mathcal{H}_{U_1}\otimes\mathcal{H}_{B_1}\otimes\mathcal{H}_{B_2}),
\end{eqnarray}with
\begin{equation}
\mathcal{W}^{-,u_1}(\rho) = \sum_{u_2\in\mathbb{Z}_2}p_U(u_2)\mathcal{W}^{u_1\oplus u_2}(\rho)\otimes\mathcal{W}^{u_2}(\rho)
\end{equation}and
\begin{equation}
\mathcal{W}^{+,u_2}(\rho) = \sum_{u_1\in\mathbb{Z}_2}p_U(u_1)\ketbra{u_1}{u_1}\otimes\mathcal{W}^{u_1\oplus u_2}(\rho)\otimes\mathcal{W}^{u_2}(\rho).
\end{equation}
\end{subequations}Before presenting channel splitting description for $N$ copies of the channel $\mathcal{W}$, some properties of the 2-fold case need to be given.

\begin{proposition}
Consider the transformation $(\mathcal{W},\mathcal{W})\to(\mathcal{W}^-,\mathcal{W}^+)$ for some channels satisfying \ref{Eq:synthesizedChannels}. Then the following rule
holds for the rates:
\begin{eqnarray}
I(\mathcal{W}^-)_\rho + I(\mathcal{W}^+)_\rho &\leq& 2I(\mathcal{W})_\rho\\
I(\mathcal{W}^+)_\rho &\geq& I(\mathcal{W})_\rho.
\end{eqnarray}
\label{Prop:PolarizationInf}
\end{proposition}

\begin{proof}
For the first statement, it is easy to see that
\begin{subequations}
\begin{eqnarray}
I(\mathcal{W}^-)_\rho + I(\mathcal{W}^+)_\rho &=& I(U_1;B_1 B_2) + I(U_2;B_1B_2U_1)\\
&\overset{(i)}{=}& I(U_1;B_1 B_2) + I(U_2;B_1B_2|U_1)\\
&\overset{(ii)}{=}& I(X_1 X_2;B_1 B_2)\\
&\overset{(iii)}{=}& I(X_1;B_1) + I(X_2;B_2) = I(X_1;B_1) + I(\mathcal{W}).
\end{eqnarray}
\end{subequations}The equality in $(i)$ follows from $U_1$ and $U_2$ being independent. $(ii)$ is derived from the chain rule for mutual information and the existence of a one-to-one function from
$U_1,U_2$ to $X_1,X_2$. Lastly, the independence between $X_1$ and $X_2$ is applied in $(iii)$. Now, notice that we are dealing with a (possibly non-uniform) Bernoulli random
variable $X_1$, which implies $I(X_1;B_1)\neq I(\mathcal{W})$ in general.
From $X_1 = U_1+U_2$ and $X_2=U_2$, we have that $X_1$ is $\text{Ber}(p^2 + (1-p)^2)$ and $X_2$ is $\text{Ber}(p)$. Then, it is possible to bound $I(X_2;B_2) - I(X_1;B_1)$
by the following inequality
\begin{eqnarray}
I(X_2;B_2) - I(X_1;B_1) &\geq& -\log\{p^2+(1-p)^2\} - \log\{p(1-p)\}-\log\text{Tr}\{\sqrt{\mathcal{W}^0(\rho)}\sqrt{\mathcal{W}^1(\rho)}\}\nonumber\\
&& -2\sqrt{(1-F(\mathcal{W}^0(\rho),\mathcal{W}^1(\rho))^2)2p(1-p)(p^2+(1-p)^2)}.
\label{Eq:diff_qmi}
\end{eqnarray}We have used the bounds
\begin{equation}
I(X_2;B_2)\geq -\log\text{Tr}\{(p_{X_2}(0)\sqrt{\mathcal{W}^0(\rho)} + p_{X_2}(1)\sqrt{\mathcal{W}^1(\rho)})^2\},
\end{equation}from \cite[Prop. 1]{Holevo2000}, and
\begin{subequations}
\begin{eqnarray}
I(X_1;B_1) &\leq& H(\sigma) = h\Big{(}\frac{1}{2}(1-\sqrt{1-4(1 - F(\mathcal{W}^0(\rho),\mathcal{W}^1(\rho))^2)2p(1-p)(p^2+(1-p)^2)})\Big{)}\\
&\leq& 2\sqrt{(1 - F(\mathcal{W}^0(\rho),\mathcal{W}^1(\rho))^2)2p(1-p)(p^2+(1-p)^2)},
\end{eqnarray}
\end{subequations}where
\begin{equation}
\sigma =
\begin{pmatrix}
p_{X_2}(0) & \sqrt{p_{X_2}(0)p_{X_2}(1)}F(\mathcal{W}^0(\rho),\mathcal{W}^1(\rho))\\
\sqrt{p_{X_2}(0)p_{X_2}(1)}F(\mathcal{W}^0(\rho),\mathcal{W}^1(\rho)) & p_{X_2}(0)
\end{pmatrix},
\end{equation}from \cite[Thm. 3]{Roga2010}.
Calling the RHS of Eq.~\ref{Eq:diff_qmi} by $f(p)$ and analyzing its first and second derivative, we can conclude that
$f(p)\geq 0$ and, thus, $I(\mathcal{W}^-)_\rho + I(\mathcal{W}^+)_\rho \leq 2I(\mathcal{W})_\rho$.

The second statement is derived from
\begin{equation}
I(\mathcal{W}^+)_\rho = I(U_2;B_1B_2U_1)\geq I(U_2;B_2)= I(X_2;B_2) = I(\mathcal{W}).
\end{equation}
\end{proof}

\begin{remark}
Proposition~\ref{Prop:PolarizationInf} shows that the polarization phenomenon can also be derived in quantum reading problematic. However, it is worth mention
the importance of the probe state used during the process. As explicitly shown in Section~\ref{sec:inputDependenceAnalysis}, there are probe states that can polarize
``faster'' than other, in the sense that $I(\mathcal{W}^+)_\rho - I(\mathcal{W}^-)_\rho > I(\mathcal{W}^+)_\sigma - I(\mathcal{W}^-)_\sigma$ for some probe states
$\rho$ and $\sigma$. Therefore, even though any probe state generates the polarization phenomenon and, thus, can be used to achieve a nonzero communication rate with arbitrarily
low error probability, it may not be optimal. There may exist another probe state that produces the same results but demanding a polar code with lower length.
\end{remark}

The following proposition shown in Ref.~\cite{Renes2015} relates the reliability of the synthesized channels to the reliability of the original one.

\begin{proposition}\cite[Prop. 4]{Renes2015}
Consider the transformation $(\mathcal{W},\mathcal{W})\to(\mathcal{W}^-,\mathcal{W}^+)$ for some channels satisfying Eq.~\ref{Eq:synthesizedChannels}. Then the following rule
holds for the reliability:
\begin{eqnarray}
Z(\mathcal{W}^+)_\rho &=& Z(\mathcal{W})_\rho^2\\
Z(\mathcal{W}^-)_\rho &\leq& 2Z(\mathcal{W})_\rho - Z(\mathcal{W})_\rho^2.
\end{eqnarray}
\end{proposition}

Now, we can extend the splitting analysis to a $N$-fold combining of $\mathcal{W}$. Wilde and Guha have shown how to extend the previous
characterization of synthesized channels~\cite{Wilde2013}. We are going to follow a similar path but using a proper description for quantum
memory cells. Let $N = 2^n$, where $n\in\mathbb{N}_0$. First of all, in this general case is preferable to label the synthesized channels by natural numbers instead of $\{+,-\}^n$.
Let $i\in\{1,2,\ldots, N\}$. The $i$-th synthesized channel is given by the map
\begin{equation}
\mathcal{W}_N^{(i)}\colon u_i\mapsto \mathcal{W}^{(i),u_i}(\rho^{\otimes N}),
\label{Eq:synthesizedChannelsGeneral}
\end{equation}where
\begin{eqnarray}
\label{Eq:DefOverlineW}
\mathcal{W}^{(i),u_i}(\rho^{\otimes N}) &=& \sum_{u_1^{i-1}\in\mathbb{Z}_2^{i-1}}p_{U^{i-1}}(u_1^{i-1})\ketbra{u_1^{i-1}}{u_1^{i-1}}\otimes\overline{\mathcal{W}}^{u_1^{i}}(\rho^{\otimes N}),\\
\overline{\mathcal{W}}^{u_1^{i}}(\rho^{\otimes N}) &=& \sum_{u_{i+1}^N\in\mathbb{Z}_2^{N-i}}p_{U_{i+1}^N}(u_{i+1}^N)\mathcal{W}^{u^N G_N}(\rho^{\otimes N}).
\label{Eq:DefOverlineW2}
\end{eqnarray}From Eq.~\ref{Eq:DefOverlineW} we see that the description of $\mathcal{W}_N^{(i)}$ supposes the knowledge of the previous input $u_1^{i-1}$.
For a finite-length analysis, this hypothesis may hold with the use of a ``genie-aided'' successive cancellation decoder similar to Refs.~\cite{arikan09,Wilde2013}.
In the asymptotic analysis, this is not needed.

Having described channel combining and splitting, and given important definitions like the meaning of synthesized channels, our next step is to characterize the behavior of rate,
reliability, and error probability in the asymptotic scenario. See the next subsection.

	\subsubsection{Rate, Reliability, and Error Probability}

The first result present in this section describes the connection between rate and reliability. It shows, as expected, that
the rate $I(\mathcal{W})_\rho\to 0$ (or $I(\mathcal{W})_\rho\to h(p)$) when the reliability $Z(\mathcal{W})_\rho\to 2\sqrt{p(1-p)}$ (or $Z(\mathcal{W})_\rho\to 0$).
%The below proposition states this result.

\begin{proposition}
Let $X\sim \text{Ber}(p)$ and $\mathcal{W}^x(\rho)$ be a quantum memory cell, where $x\in\mathbb{Z}_2$.
Then the following holds
\begin{eqnarray}
I(\mathcal{W})_\rho &\geq& h(p) - \log(1+Z(\mathcal{W})_\rho),\\
I(\mathcal{W})_\rho &\leq& \sqrt{4p(1-p)-Z(\mathcal{W})_\rho^2}.
\end{eqnarray}
\label{Prop:MutualInformationAndReliability}
\end{proposition}

\begin{proof}
The first inequality follows from \cite[Thm. C.1]{Hirche:Thesis}.
For the second one, we need to use the following inequality derived from Theorem 3 in Ref.~\cite{Roga2010}:
\begin{equation}
I(\mathcal{W})_\rho \leq H(\sigma),
\end{equation}where
\begin{equation}
\sigma =
\begin{pmatrix}
p & \frac{Z(\mathcal{W})_\rho}{2}\\
\frac{Z(\mathcal{W})_\rho}{2} & 1-p
\end{pmatrix}.
\end{equation}Thus
\begin{eqnarray}
I(\mathcal{W})_\rho &\leq& h\Big{(}\frac{1}{2}(1 - \sqrt{1 - 4p(1-p)+Z(\mathcal{W})_\rho^2})\Big{)}\\
			   &\leq& \sqrt{4p(1-p)-Z(\mathcal{W})_\rho^2}.
\end{eqnarray}
\end{proof}

Before presenting the next result, we need to introduce a new classical-quantum symmetric channel. Let $\tilde{U}_i\sim\text{Ber}(\frac{1}{2})$,
$U_j\sim\text{Ber}(p)$ for all $i,j=1, \ldots,N$. Assume that the output system of the classical-quantum channel $\tilde{\mathcal{W}}$ is given by
$\mathcal{D}(\mathcal{H}_{\tilde{Z}}\otimes\mathcal{H}_{B})$, where $\sigma\in\mathcal{D}(\mathcal{H}_{\tilde{Z}}\otimes\mathcal{H}_{B})$ can be written as
\begin{equation}
\sigma = \sum_{\tilde{z}\in\mathbb{Z}_2}\ketbra{\tilde{z}}{\tilde{z}}^{\tilde{Z}}\otimes \rho^B.
\end{equation}Notice that $\tilde{Z}$ represents the classical input system in the same way as the system $Z$ associated with the quantum channel $\mathcal{W}$. 
Therefore, the associated quantum systems are the same. 
We keep different notation here to emphasize that one is associated with a symmetric channel and the other with an asymmetric channel. Now, we can properly describe the 
classical-quantum channel $\tilde{\mathcal{W}}$. Fix $\rho\in\mathcal{D}(\mathcal{H}_A)$ and let $u\in\mathbb{Z}_2$. Then $\tilde{\mathcal{W}}$ is the map given by
\begin{align}
  \tilde{\mathcal{W}} \colon \tilde{U} &\to \mathcal{D}(\mathcal{H}_{\tilde{Z}}\otimes\mathcal{H}_{B})\nonumber\\
  \tilde{u} &\mapsto \tilde{\mathcal{W}}^{\tilde{u}}(\rho) := \sum_{u\in \mathcal{U}}p_U(u)\ketbra{\tilde{u}\oplus u}{\tilde{u}\oplus u}^{\tilde{Z}}\otimes \mathcal{W}^u (\rho)^B.
  \label{Eq:Symmetric_Channel}
\end{align}It is easy to see that $\tilde{\mathcal{W}}$ is a symmetric classical-quantum channel. Indeed, this follows from the uniform distribution of $\tilde{U}$ and the construction of the
$\mathcal{H}_{\tilde{Z}}$ part. The next step is to describe the synthesized channels generated in polar coding on $\tilde{\mathcal{W}}$.
Following the characterization given in Eq.~\ref{Eq:synthesizedChannelsGeneral}, we have that
\begin{equation}
\tilde{\mathcal{W}}_N^{(i),\tilde{u}_i}(\rho^N) = \sum_{\tilde{u}_1^{i-1}}\frac{1}{2^{i-1}}\ketbra{\tilde{u}_1^{i-1}}{\tilde{u}_1^{i-1}}\otimes\Big{(}\sum_{u^N}p_{U^N}(u^N)\mathcal{W}^{u^N G_N}(\rho^N)\otimes\Big{(}\sum_{\tilde{u}_{i+1}^N}\frac{1}{2^{N-i}}\ketbra{(\tilde{u}^N\oplus u^N )G_N}{(\tilde{u}^N\oplus u^N )G_N}\Big{)}\Big{)}.
\end{equation}With these tools, we show in the next proposition that any asymmetric classical-quantum channel can be described via a symmetric one. Moreover, in the following,
a relation between the reliabilities of these channels is provided.

\begin{proposition}
Let $\tilde{\mathcal{W}}_N^{(i)}\colon \tilde{U}_i\rightarrow \mathcal{D}(\mathcal{H}_{\tilde{U}_1^{i-1}}\otimes\mathcal{H}_{B^N}\otimes\mathcal{H}_{\tilde{Z}^N})$, where $\tilde{U}_i\sim\text{Ber}(\frac{1}{2})$ for all $i=1, \ldots,N$. Then,

\begin{equation}
\mathcal{W}_N^{(i),u_i}(\rho^N) = \text{Tr}_{\tilde{Z}^N}\{\tilde{\mathcal{W}}_N^{(i),u_i}(\rho^N)\ketbra{0}{0}^{Z^N}\}.
\end{equation}
\label{Eq:Asymmetric_Symmetric_Correspondency}
\end{proposition}

\begin{proof}
Let $\tilde{\rho}^{(i)}$ be the joint input-output density operator of $\tilde{\mathcal{W}}_N^{(i),\tilde{u}_i}$. Thus,
\begin{subequations}
\begin{eqnarray}
\tilde{\rho}^{(i)} &=& \frac{1}{2} \sum_{\tilde{u}_i\in\mathbb{Z}_2}\ketbra{\tilde{u}_i}{\tilde{u}_i}\otimes\tilde{\mathcal{W}}_N^{(i),\tilde{u}_i}(\rho^N)\\
&=&\sum_{\tilde{u}^N}\frac{1}{2^N}\ketbra{\tilde{u}_1^{i}}{\tilde{u}_1^{i}}\otimes\Big{(}\sum_{u^N}p_{U^N}(u^N)\mathcal{W}^{u^N G_N}(\rho^N)\otimes\ketbra{(\tilde{u}^N\oplus u^N )G_N}{(\tilde{u}^N\oplus u^N )G_N}\Big{)}\\
&=& \sum_{u^N}p_{U^N}(u^N)\mathcal{W}^{u^N G_N}(\rho^N)\otimes\Big{(}\sum_{\tilde{u}^N}\frac{1}{2^N}\ketbra{\tilde{u}_1^{i}}{\tilde{u}_1^{i}}\otimes\ketbra{(\tilde{u}^N\oplus u^N )G_N}{(\tilde{u}^N\oplus u^N )G_N}\Big{)}.
\end{eqnarray}
\end{subequations}Defining $\tilde{z}^N = (\tilde{u}^N\oplus u^N )G_N$, it is possible to rearrange the sums as
\begin{eqnarray}
\tilde{\rho}^{(i)} &=& \sum_{u^N}p_{U^N}(u^N)\mathcal{W}^{u^N G_N}(\rho^N)\otimes\Big{(}\sum_{\substack{\tilde{z}^N\\ [u^N\oplus\tilde{z}^N G_N]_1^i=\tilde{u}_1^i}}\frac{1}{2^N}\ketbra{[u^N\oplus\tilde{z}^N G_N]_1^i}{[u^N\oplus\tilde{z}^N G_N]_1^i}\otimes\ketbra{\tilde{z}^N}{\tilde{z}^N}\Big{)}\nonumber\\
&=& \sum_{\tilde{z}^N}\frac{1}{2^N}\ketbra{\tilde{z}^N}{\tilde{z}^N}\otimes\Big{(}\sum_{\tilde{u}_1^i}p_{U_1^i}(\tilde{u}_1^i\oplus[\tilde{z}^N G_N]_1^i)\ketbra{\tilde{u}_1^i}{\tilde{u}_1^i}\otimes\mathcal{W}^{ [(\tilde{u}_1^i,0_{i+1}^N)G_N]_{1}^i \oplus\tilde{z}_1^i }(\rho^i)\nonumber\\
&\otimes&\Big{(}\sum_{u_{i+1}^N}p_{U_{i+1}^N}(u_{i+1}^N)\mathcal{W}^{[(0_1^i,u_{i+1}^N)G_N]_{i+1}^N}(\rho^{N-i})\Big{)}\Big{)}.
\label{Eq:jointProbeSymmetricChannel}
\end{eqnarray}Now, we can see that
\begin{equation}
\text{Tr}_{\tilde{Z}^N}\{\tilde{\rho}^{(i)}\ketbra{0}{0}^{Z^N}\} = \sum_{u_1^i}p_{U_1^i}(u_1^i)\ketbra{u_1^i}{u_1^i}\otimes\Big{(}\sum_{u_{i+1}^N}p_{U_{i+1}^N}(u_{i+1}^N)\mathcal{W}^{u^N G_N}(\rho^N)\Big{)},
\end{equation}which is the joint input-output state of $\mathcal{W}^{(i),u_i}_N(\rho^N)$.
\end{proof}

\begin{proposition}
Let $\mathcal{W}_N^{(i)}$ and $\tilde{\mathcal{W}}_N^{(i)}$ be the synthesized quantum channels described in Eq.~\ref{Eq:synthesizedChannelsGeneral} and Proposition~\ref{Eq:Asymmetric_Symmetric_Correspondency}, respectively. Then
\begin{equation}
Z(\tilde{\mathcal{W}}_N^{(i)})_\rho = Z(\mathcal{W}_N^{(i)})_\rho.
\end{equation}
\end{proposition}

\begin{proof}
From the definition of reliability for symmetric classical-quantum channels in Ref.~\cite{Wilde2013} and Eq~\ref{Eq:jointProbeSymmetricChannel}, we have that
\begin{eqnarray}
Z(\tilde{\mathcal{W}}_N^{(i)})_\rho &=& F(\tilde{\mathcal{W}}_N^{(i),0}(\rho^N), \tilde{\mathcal{W}}_N^{(i),1}(\rho^N))\\
&=& \sum_{\tilde{z}^N}\sum_{\tilde{u}_1^{i-1}}\frac{p_{U_1^{i-1}}(\tilde{u}_1^{i-1}\oplus[\tilde{z}^N G_N]_1^{i-1})\sqrt{p_U(0\oplus[\tilde{z}^N G_N]_i)p_U(1\oplus[\tilde{z}^N G_N]_i)}}{2^{N-1}}
F(w^{(i)}_0,w^{(i)}_1),
\end{eqnarray}where
\begin{eqnarray}
w^{(i)}_0 &=& \mathcal{W}^{ [(\tilde{u}_1^{i-1},0_{i}^N)G_N]_{1}^i \oplus\tilde{z}_1^i }(\rho^i)\sum_{u_{i+1}^N}p_{U_{i+1}^N}(u_{i+1}^N)\mathcal{W}^{[(0_1^i,u_{i+1}^N)G_N]_{i+1}^N}(\rho^{N-i}),\\
w^{(i)}_1 &=& \mathcal{W}^{ [(\tilde{u}_1^{i-1},1,0_{i+1}^N)G_N]_{1}^i \oplus\tilde{z}_1^i }(\rho^i)\sum_{u_{i+1}^N}p_{U_{i+1}^N}(u_{i+1}^N)\mathcal{W}^{[(0_1^i,u_{i+1}^N)G_N]_{i+1}^N}(\rho^{N-i}).
\end{eqnarray}Let $z^N = \tilde{z}^N G_N$ and $u_1^{i-1} = \tilde{u}_1^{i-1}\oplus[\tilde{z}^N G_N]_1^{i-1} = \tilde{u}_1^{i-1} \oplus z_1^{i-1}$. Since the map
$(\tilde{z}_1^N,\tilde{u}_1^{i-1})\mapsto(z_1^N,u_1^{i-1})$ is a bijection, we see that
\begin{eqnarray}
Z(\tilde{\mathcal{W}}_N^{(i)})_\rho &=& \sum_{z^N}\sum_{u_1^{i-1}}\frac{p_{U_1^{i-1}}(u_1^{i-1})\sqrt{p_U(0)p_U(1)}}{2^{N-1}}F(\mathcal{W}^{(u_1^{i-1},0)}(\rho^i)\sum_{u_{i+1}^N}p_{U_{i+1}^N}(u_{i+1}^N)\mathcal{W}^{[(0_1^i,u_{i+1}^N)G_N]_{i+1}^N}(\rho^{N-i}),\nonumber\\
& &\mathcal{W}^{(u_1^{i-1},1)}(\rho^i)\sum_{u_{i+1}^N}p_{U_{i+1}^N}(u_{i+1}^N)\mathcal{W}^{[(0_1^i,u_{i+1}^N)G_N]_{i+1}^N}(\rho^{N-i}))\nonumber\\
&=& 2\sqrt{p_U(0)p_U(1)}\sum_{u_1^{i-1}}p_{U_1^{i-1}}(u_1^{i-1})F(\overline{\mathcal{W}}^{(u_1^{i-1},0)}(\rho^{\otimes N}),\overline{\mathcal{W}}^{(u_1^{i-1},1)}(\rho^{\otimes N}))\nonumber\\
&=& 2\sqrt{p_U(0)p_U(1)}F(\mathcal{W}^{(i),0}(\rho^{\otimes N}),\mathcal{W}^{(i),1}(\rho^{\otimes N}))\nonumber\\
&=& Z(\mathcal{W}_N^{(i)})_\rho,
\end{eqnarray}where $\overline{\mathcal{W}}^{u_1^i}(\rho^{\otimes N})$ is defined in Eq.~\ref{Eq:DefOverlineW}.
\end{proof}

The next theorem shows that the asymptotically fraction of good channels is equals to the mutual information between the classical system $X$ and the quantum system $B$. 

\begin{theorem}
Let $\mathcal{W}_N^{(i)}$ be the synthesized quantum channels described in Eq.~\ref{Eq:synthesizedChannelsGeneral} and $Z(X|Y)$ be the reliability of two Bernoulli random variables $X$ and $Y$; i.e., 
$Z(X|Y) = 2\sum_y\sqrt{p_{X,Y}(0,y)p_{X,Y}(1,y)}$. For every $\beta<1/2$,
we have
\begin{eqnarray}
\lim_{n\rightarrow\infty}\frac{1}{2^n}|\{i\in\mathbb{Z}_2^n\colon Z(\mathcal{W}_N^{(i)})_\rho\leq 2^{-2^{n\beta}} \text{ and }Z(U_i|U_1^{i-1})\geq 1 - 2^{-2^{n\beta}}\}| &=& I(X;B)_\rho,\\
\lim_{n\rightarrow\infty}\frac{1}{2^n}|\{i\in\mathbb{Z}_2^n\colon Z(\mathcal{W}_N^{(i)})_\rho\geq 1 - 2^{-2^{n\beta}} \text{ and }Z(U_i|U_1^{i-1})\leq 2^{-2^{n\beta}}\}| &=& 1 - I(X;B)_\rho,
\end{eqnarray}
\label{Theorem:Rate}
\end{theorem}
\begin{proof}
Firstly, notice that
\begin{eqnarray}
I(\tilde{\mathcal{W}}) &=& I(\tilde{X}; \tilde{X}\oplus X,B)\nonumber\\
&=& H(\tilde{X}\oplus X,B) - H(\tilde{X}\oplus X,B|\tilde{X})\nonumber\\
&=& 1 + H(B) - H(X,B) = 1-H(X|B).
\label{Eq:mutualInformationOfSymmetricChannel}
\end{eqnarray}Applying Eq.~\ref{Eq:mutualInformationOfSymmetricChannel} to Proposition~\ref{Proposition:Preliminaries}, we can deduce that
\begin{subequations}
\label{Eq1:ProofTheorem}
\begin{eqnarray}
\lim_{n\rightarrow\infty}\frac{1}{2^n}|\{i\in\mathbb{Z}_2^n\colon Z(\mathcal{W}_N^{(i)})_\rho\leq 2^{-2^{n\beta}}\}| &=& 1 - H(X|B)_\rho,\\
\lim_{n\rightarrow\infty}\frac{1}{2^n}|\{i\in\mathbb{Z}_2^n\colon Z(\mathcal{W}_N^{(i)})_\rho\geq 1 - 2^{-2^{n\beta}}\}| &=& H(X|B)_\rho.
\end{eqnarray}
\end{subequations}

Additionally, it is possible to derive $Z(\mathcal{W}_N^{(i)})_\rho = Z(U_i|U_1^{i-1})$ and $H(X|B) = H(X)$ if the output probe state of
$\mathcal{W}$ is independent of $X$. Thus, a similar result to Eq.~\ref{Eq1:ProofTheorem} is derived
\begin{subequations}
\begin{eqnarray}
\lim_{n\rightarrow\infty}\frac{1}{2^n}|\{i\in\mathbb{Z}_2^n\colon Z(U_i|U_1^{i-1})\leq 2^{-2^{n\beta}} \}| &=& 1 - H(X),\\
\lim_{n\rightarrow\infty}\frac{1}{2^n}|\{i\in\mathbb{Z}_2^n\colon Z(U_i|U_1^{i-1})\geq 1 - 2^{-2^{n\beta}}\}| &=& H(X).
\end{eqnarray}
\end{subequations}Now, define $A, B, C,$ and $D$ as the sets
\begin{subequations}
\begin{eqnarray}
A &=& \{i\colon Z(\mathcal{W}_N^{(i)})_\rho\leq 2^{-2^{n\beta}}\},\\
B &=& \{i\colon Z(\mathcal{W}_N^{(i)})_\rho\geq 1 - 2^{-2^{n\beta}}\},\\
C &=& \{i\colon Z(U_i|U_1^{i-1})\leq 2^{-2^{n\beta}}\},\\
D &=& \{i\colon Z(U_i|U_1^{i-1})\geq 1 - 2^{-2^{n\beta}}\}.
\end{eqnarray}
\end{subequations}From Proposition~\ref{Prop:MutualInformationAndReliability} and $H(U_i|U_1^{i-1},B^N)\leq H(U_i|U_1^{i-1})$, it is
possible to see that $B\cap C$ is empty for sufficiently large $n$. Furthermore,
\begin{equation}
\lim_{n\rightarrow\infty}\frac{|A|+|B|}{2^n} = \lim_{n\rightarrow\infty}\frac{|C|+|D|}{2^n} = 1.
\end{equation}Therefore, the claim is derived from
\begin{equation}
\lim_{n\rightarrow\infty}\frac{|B\cup C|}{2^n} = \lim_{n\rightarrow\infty}\frac{|B|+|C|}{2^n} = 1 - I(X;B)_\rho.
\end{equation}and
\begin{equation}
\lim_{n\rightarrow\infty}\frac{|A\cap D|}{2^n} = 1 - \lim_{n\rightarrow\infty}\frac{|B\cup C|}{2^n} = I(X;B)_\rho.
\end{equation}
\end{proof}

As can be seen in the theorem statement and elaborated in the proof, we had to impose an additional constraint on the reliability of 
$U_i$ given the previous $U_1^{i-1}$ in order to derive our result. 
This is because we are dealing with asymmetric quantum reading. For symmetric quantum reading, the constrain is not needed.

\begin{remark}
Subsections~\ref{subsec:channel_combining} and \ref{subsec:channel_splitting}, and Proposition~\ref{Prop:MutualInformationAndReliability} are easily applied to
symmetric quantum channels discrimination, having it as a particular case when $p_X(0)=p_X(1) = \frac{1}{2}$. For the result in Theorem~\ref{Theorem:Rate}, there is
no need to introduce the new type of channel and, furthermore, the asymptotic analysis of the rate of good and bad channels does not impose anything on $Z(U_i|U_1^{i-1})$.
The symmetric quantum reading treatment of Theorem~\ref{Theorem:Rate} goes similarly to Section IV of Ref.~\cite{Wilde2013}.
\end{remark}

\subsection{Polar Coding and Decoding Scheme}

The encoding protocol of polar codes consists of setting up the labels for information bits and frozen bits. The former label is denoted by
$\mathcal{A}$ and the later by $\mathcal{A}^c$. Thus, we use bits $u_\mathcal{A} = \{u_i\}_{i\in\mathcal{A}}$ to transmit information. On the frozen bits,
they can be fixed for the whole transmission or can depend on the previous $u_1^{i-1}$ bits. For asymmetric channels, as our case, the latter strategy is more
suitable. See Ref.~\cite{honda13} for more explanations.

The construction of a codeword is done as follows. The source generates a uniform sequence $u_1^{|\mathcal{A}|}$. Next, the encoder determines the value
$u_i$, $i\in\mathcal{A}^c$, of the frozen bits in the ascending order by $u_i = \lambda_i(u_1^{i-1})$, where $\lambda_i$ is a function from $\{0,1\}^{i-1}$ to
$\{0,1\}$. Putting the sequence $u_1^{|\mathcal{A}|}$ in the information bits $u_\mathcal{A}$ and the frozen bits in the remaining coordinates, we have the vector $u^N$.
Now, the codeword is given by $x^N = u^N G_N$. It is clear that the code rate is $R = |\mathcal{A}|/N$. It remains to describe how the set $\mathcal{A}$ is
determined and which functions $\lambda_i$ do we use.

Determining the set of information bits is crucial for the performance of a polar code. Firstly, for a sufficiently large code length $N$, we choose the cardinality of
$\mathcal{A}$ so that the code rate $R < I(\mathcal{W})_\rho$. Next, we have to choose coordinates to constitute the set $\mathcal{A}$. A common approach is to select the coordinates
with smallest reliabilities. Formally, this selection goes as follows. Let $\mathcal{A}$ be the set that for any $j\in\mathcal{A}^c$ we have $Z(\mathcal{W}_N^{(j)})\geq Z(\mathcal{W}_N^{(i)})$
for all $i\in\mathcal{A}$. Additionally, because the construction under consideration is for asymmetric channels, we also impose $Z(U_i|U_1^{i-1})$, for all $i\in\mathcal{A}$,
to be large when compared with the elements in $\mathcal{A}^c$. Theorem~\ref{Theorem:Rate} makes use of this constrain to characterize the asymptotic rates in polar coding.
If any of these two constrains is not satisfied, we decrease the cardinality of $\mathcal{A}$ down to when they are satisfied.

The functions $\lambda_i$ used in this paper is such that optimize the probability of frozen bit output. Let $\Lambda_{\mathcal{A}^c} = \{\Lambda_i\}_{i\in\mathcal{A}^c}$ be
random variables which are independent of each other and input and output systems, and satisfy
\begin{equation}
p_{\Lambda_i}[\Lambda_i(u_1^{i-1})] = p_{U_i|U_1^{i-1}}(1|u_1^{i-1}),
\end{equation}for all $u_1^{i-1}\in\{0,1\}^{i-1}$. Then, $\lambda_i$ is a realization of the random variable $\Lambda_i$. There are practical methods to generate in practice the
functions $\lambda_i$ using pseudorandom number generators~\cite{honda13}, but there is no need to address it in this paper.

Considering symmetric channel discrimination, a similar encoding scheme can be proposed. First of all, the algorithm for choosing the frozen bits is exactly the one presented here
without the dependence of the previous bits. Thus, the encoding map is a realization of the random variable $\tilde{\Lambda}_i$, where
\begin{equation}
p_{\Lambda_i}[\Lambda_i = i] = \frac{1}{2},
\end{equation}for $i=1,2$. Secondly, the proposed strategy for defining the set $\mathcal{A}$ is applicable for symmetric channels discrimination. We only need to drop the constrain
on $Z(U_i|U_1^{i-1})$. Approximation techniques in symmetric polar codes \cite{Tal:2013} can compose the decoding scheme in order to derive a faster encoding scheme.

Now we describe the decoding process. Suppose the sequential decoder has obtained, up to this moment, the vector $\hat{u}_1^{i-1}$ and plans to obtain $\hat{u}_i$.
Then the decoding process divides into two cases. For the coordinates in $\mathcal{A}^c$, we apply the inverse encoding function depending on the previous
coordinates. Namely, we employ
\begin{equation}
\hat{u}_i = \lambda_i^{-1}(\hat{u}_1^{i-1}).
\end{equation}Notice that no measurement is implemented in this step and, for sufficiently large $N$, the error probability is arbitrarily low.
See Theorem~\ref{Theorem:Rate} for the proof. Next, we deal with the information bits. Quantum successive cancellation decoder is used in here. First of all, this makes our
decoding strategy and analysis totally different from any decoding method used for classical polar codes. Secondly, because of the constructive approach adopted in here, with
measurements created for our specific task, the polar decoding strategy differs from previous works on quantum polar codes.

For characterizing the error probability decay of the polar codes constructed, we firstly show the existence of ``pretty good measurements'' design to decode quantum memory cell and having
desirable error probabilities in Proposition~\ref{Proposition:PrettyGoodMeas}.

\begin{proposition}
Let $\mathcal{W}\colon x\in\mathcal{X}\rightarrow \rho_x^{UB}\in\mathcal{D}(\mathcal{H}_{UB})$ be a cq channel such that
\begin{equation}
\rho_x^{UB} = \sum_{u\in\mathcal{U}}p_{U}(u) \ketbra{u}{u}^U\otimes\rho_{x,u}^B,
\label{Proposition13:Defrhox}
\end{equation}where $X$ is a discrete random variable, and $\{\ket{u}\}_{u\in\mathcal{U}}$ is an orthonormal basis of a finite-dimensional Hilbert space. Then, there exists a POVM
$\{\Lambda_{x}^{UB}\}_{x\in\mathcal{X}}$ satisfying
\begin{equation}
1 - \sum_{x\in\mathcal{X}}p_X(x)\text{Tr}\{\Lambda_x\rho_x\} < \frac{1}{2}Z(\mathcal{W}).
\end{equation}
\label{Proposition:PrettyGoodMeas}
\end{proposition}

\begin{proof}
To derive our claim, we need to invoke a result from Barnum and Knill \cite{Barnum2002}. Let $p_U(u)$ be the probability of a system to be found in the state $\rho_u$, for $u\in\mathcal{U}$, then
there exists a POVM $\{\Lambda_u^* \}_{u\in\mathcal{U}}$ such that the average success probability is lower bounded as
\begin{equation}
P_{succ} = \sum_{u\in\mathcal{U}} p_U(u)\text{Tr}\{\rho_u^{\otimes M}\Lambda_u^*\} \geq 1 - \sum_{u\neq v}\sqrt{p_U(u) p_U(v)}F(\rho_u^M,\rho_v^M).
\end{equation}In particular, there exists a POVM $\{\Lambda_u\}_{u\in\mathcal{U}}$ such that the probability of error follows the following inequality
\begin{equation}
P_{err} = 1 - \sum_{u\in\mathcal{U}} p_U(u)\text{Tr}\{\rho_u\Lambda_u\} \leq \sum_{u\neq v}\sqrt{p_U(u) p_U(v)}F(\rho_u,\rho_v).
\label{Proposition13:Eq73}
\end{equation}Now, we return to our proof. Assume that $\mathcal{W}^{u}\colon x\rightarrow\rho_{x,u}$ is a cq channel, where $x\in\mathcal{X}$
is a realization of a discrete random variable $X$. The result in Eq.~\ref{Proposition13:Eq73} says that there exist POVMs $\{\Lambda_{x,u}\}_{x\in\mathcal{X}}$ satisfying
\begin{equation}
1 - \sum_{x\in\mathcal{X}} p_X(x)\text{Tr}\{\rho_{x,u}\Lambda_{x,u}\} < \frac{1}{2}Z(\mathcal{W}^{u}).
\label{Proposition13:InequalityZ}
\end{equation}Defining, for every $u\in\mathcal{U}$, the POVM
\begin{equation}
\Lambda_x := \sum_{u\in\mathcal{U}}\Lambda_{x,u}\otimes\ketbra{u}{u},
\label{Proposition13:DefPix}
\end{equation}we have that
\begin{eqnarray}
&1& - \sum_{x\in\mathcal{X}}p_X(x)\text{Tr}\{\Lambda_x\rho_x\} \overset{(i)}{=} 1 - \sum_{x\in\mathcal{X}}p_X(x)\sum_{u\in\mathcal{U}}p_U(u)\text{Tr}\{\Lambda_{x,u}\rho_{x,u}\}\\
&=&\sum_{u\in\mathcal{U}}p_{U}(u)\sum_{x\in\mathcal{X}}p_X(x)\Big{(}1 - \text{Tr}\{\Lambda_{x,u}\rho_{x,u}\}\Big{)}\\
&\overset{(ii)}{<}& \sum_{u\in\mathcal{U}}p_U(u)\frac{1}{2}Z(\mathcal{W}^{u})\\
&\overset{(iii)}{=}& \sum_{\substack{x,x'\in\mathcal{X}\\x\neq x'}}\sqrt{p_X(x)p_X(x')}F\Big{(}\sum_{u\in\mathcal{U}}p_U(u)\ketbra{u}{u}\otimes\rho_{x,u},\sum_{u\in\mathcal{U}}p_U(u)\ketbra{u}{u}\otimes\rho_{x',u}\Big{)}\\
&\overset{(iv)}{=}& \sum_{\substack{x,x'\in\mathcal{X}\\x\neq x'}}\sqrt{p_X(x)p_X(x')}F\Big{(}\rho_x,\rho_{x'}\Big{)}\\
&=& \frac{1}{2}Z(\mathcal{W}).
\end{eqnarray}The equality in $(i)$ follows from the definitions of $\rho_x$ and $\Lambda_x$ in Eq.~\ref{Proposition13:Defrhox} and \ref{Proposition13:DefPix}, respectively.
The inequality in Eq.~\ref{Proposition13:InequalityZ} is used to obtain $(ii)$. The item $(iii)$ is just the computation of the reliability for $\mathcal{W}^u$. Lastly, we use in $(iv)$ the
fact that $F(\sum_x p_X(x)\ketbra{x}{x}\otimes \sigma_x,\sum_x q_X(x)\ketbra{x}{x}\otimes \tau_x) = \sum_x\sqrt{p_X(x)q_X(x)}F(\sigma_x,\tau_x)$.
\end{proof}

Analyzing a quantum successive cancellation decoder employing the measurements in Proposition~\ref{Proposition:PrettyGoodMeas} needs the use of a quantum union bound.
Lemma~\ref{Lemma:QuantumUnionBound} introduces the one used in this paper. As it is shown in the proof of Theorem~\ref{Theorem:ErrorProbDecay}, the inequality
in Lemma~\ref{Lemma:QuantumUnionBound} allows us to derive a bound for the error probability in accordance to what is expected using polar codes.

\begin{lemma}\cite[Lemma 4.1]{Oskouei2019}
Let $\rho$ be a positive semi-definite operator acting on a separable Hilbert space $\mathcal{H}_B$, let $\{\Lambda_i\}_{i=1}^L$
denote a set of positive semi-definite operators such that $0\leq \Lambda_i\leq I$ for all $i\in\{1,\ldots, L\}$, and let $c>0$. Then
the following quantum union bound holds
\begin{eqnarray}
\text{Tr}\{\rho\} - \text{Tr}\{\Pi_{\Lambda_L}\cdots\Pi_{\Lambda_1}(\rho\otimes\ketbra{\overline{0}}{\overline{0}}_{P^L})\Pi_{\Lambda_1}\cdots\Pi_{\Lambda_L}\}\leq (1+c)\text{Tr}\{(I-\Lambda_L)\rho\}\nonumber\\
+ (2+c+c^{-1})\sum_{i=2}^{L-1}\text{Tr}\{(I-\Lambda_i)\rho\} + (2+c^{-1})\text{Tr}\{(I-\Lambda_1)\rho\},
\label{Eq:QuantumUnionBound}
\end{eqnarray}where $\ket{\overline{0}_{P_L}}\equiv \ket{0_{P_1}}\otimes\cdots\otimes\ket{0_{P_L}}$ is an auxiliary state of $L$ qubit probe systems and $\Pi_{\Lambda_i}$ is a projector
defined as $\Pi_{\Lambda_i} := U_i^\dagger P_i U_i$ for some unitary $U_i$ and projector $P_i$ such that
$\text{Tr}\{\Pi_{\Lambda_i}(\rho\otimes\ketbra{\overline{0}}{\overline{0}}_{P^L})\} = \text{Tr}\{\Lambda_i\rho\}$. In particular, for any set of positive semi-definite operators, there exists
a set of projectors satisfying Eq.~\ref{Eq:QuantumUnionBound}.
\label{Lemma:QuantumUnionBound}
\end{lemma}

An important attribute of polar codes is the error probability decay when the length of the code grows. We show in Theorem~\ref{Theorem:ErrorProbDecay} that
error probability decays exponentially fast in the code length. This result motivates the use of polar codes in practical protocols devoted to discriminate quantum memory cell.

\begin{theorem}
Let $\{\mathcal{W}^x\}_{x\in\mathcal{X}}$ be a quantum memory cell and $\gamma\in\mathbb{R}$ a positive constant. Then there exists a polar code with parameters $(N,K,\mathcal{A})$ such that the error
probability is bounded above by
\begin{equation}
P_e(N,K,\mathcal{A}) \leq \gamma\sum_{i\in\mathcal{A}} Z_\rho(\mathcal{W}^{(i)}).
\end{equation}In particular, for any fixed $R = K/N< I(W)$ and $\beta<1/2$, block error probability for polar coding under sequential
	decoding satisfies
	\begin{equation}
	P_e(N,K)= o(2^{-N^\beta}),
	\end{equation}where $o(\cdot)$ is the little-O notation from complexity theory.
\label{Theorem:ErrorProbDecay}
\end{theorem}

\begin{proof}
First of all, suppose we have access to an auxiliary system $\ket{\overline{0}_{P_N}} = \ket{0_{P_1}}\otimes\cdots\otimes\ket{0_{P_N}}$. Lemma~\ref{Lemma:QuantumUnionBound}
guarantees the existence of projective measurements obtained by extending the pretty good measurements of Proposition~\ref{Proposition:PrettyGoodMeas}. We denote the projective
and pretty good measurements by the sets $\{\Pi_{\Lambda_i}\}_{i=1}^N$ and $\{\Lambda_i\}_{i=1}^N$, respectively. Then, using these
projective measurements on a quantum successive cancellation strategy for the information bits, a bound for the error probability is obtained as follows
\begin{eqnarray}
P_e(N,K,\mathcal{A}) &=& 1 - \sum_{u^N}p_{U^N}(u^N) \text{Tr}\{\Pi_{\Lambda_{u^N}}\mathcal{W}^{u^N}(\rho^N)\otimes\ketbra{\overline{0}_{P_N}}{\overline{0}_{P_N}}\}\\
&=& 1 - \sum_{u^N}p_{U^N}(u^N) \text{Tr}\{\Pi_{\Lambda_{u_1^{N-1}u_N}}^{B^N}\cdots\Pi_{\Lambda_{u_1}}^{B^N}\mathcal{W}^{u^N}(\rho^N)\otimes\ketbra{\overline{0}_{P_N}}{\overline{0}_{P_N}}\Pi_{\Lambda_{u_1}}^{B^N}\cdots\Pi_{\Lambda_{u_1^{N-1}u_N}}^{B^N}\}\\
&\leq& \sum_{u^N}p_{U^N}(u^N)\Big{(} (1+c)\text{Tr}\Big{\{}(I-\Lambda_{u_1^{N-1}u_N}^{B^N})\mathcal{W}^{u^N}(\rho^N)\Big{\}}\nonumber\\
&+& (2+c+c^{-1})\sum_{i=2}^{N-1}\text{Tr}\Big{\{}(I-\Lambda_{u_1^{i-1}u_i}^{B^N})\mathcal{W}^{u^N}(\rho^N)\Big{\}}\nonumber\\
&+& (2+c^{-1})\text{Tr}\Big{\{}(I-\Lambda_{u_1}^{B^N})\mathcal{W}^{u^N}(\rho^N)\Big{\}}\Big{)},
\label{Eq:79}
\end{eqnarray}where the second inequality is obtained from the definition of a quantum successive cancellation decoder and the inequality from Lemma~\ref{Lemma:QuantumUnionBound}.
Distributing and rearranging the sums, the following equality holds
\begin{eqnarray}
\text{RHS of Eq.~\ref{Eq:79}}&=& (1+c)\sum_{u^N}p_{U^N}(u^N) \text{Tr}\Big{\{}(I-\Lambda_{u_1^{N-1}u_N}^{B^N})\mathcal{W}^{u^N}(\rho^N)\Big{\}}\nonumber\\
&+& (2+c+c^{-1})\sum_{i=2}^{N-1}p_U(u_i)\sum_{u_{1}^{i-1}}p_{U_{1}^{i-1}}(u_{1}^{i-1})\text{Tr}\Big{\{}(I-\Lambda_{u_1^{i-1}u_i}^{B^N})\sum_{u_{1}^{i-1}}p_{U_{i+1}^{N}}(u_{i+1}^{N})\mathcal{W}^{u^N}(\rho^N)\Big{\}}\nonumber\\
&+& (2+c^{-1})\sum_{u_{1}}p_{U}(u_1)\text{Tr}\Big{\{}(I-\Lambda_{u_1}^{B^N})\sum_{u_{2}^{i-1}}p_{U_{2}^{N}}(u_{2}^{N})\mathcal{W}^{u^N}(\rho^N)\Big{\}}.
\label{Eq:80}
\end{eqnarray}Continuing,
\begin{eqnarray}
\text{RHS of Eq.~\ref{Eq:80}}&=& (1+c)\sum_{u^N}p_{U^N}(u^N) \text{Tr}\Big{\{}(I-\Lambda_{u_1^{N-1}u_N}^{B^N})\overline{\mathcal{W}}^{u_1^N}(\rho^N)\Big{\}}\nonumber\\
&+& (2+c+c^{-1})\sum_{i=2}^{N-1}p_U(u_i)\sum_{u_{1}^{i-1}}p_{U_{1}^{i-1}}(u_{1}^{i-1})\text{Tr}\Big{\{}(I-\Lambda_{u_1^{i-1}u_i}^{B^N})\overline{\mathcal{W}}^{u_1^i}(\rho^N)\Big{\}}\nonumber\\
&+& (2+c^{-1})\sum_{u_{1}}p_{U}(u_1)\text{Tr}\Big{\{}(I-\Lambda_{u_1}^{B^N})\overline{\mathcal{W}}^{u_1}(\rho^N)\Big{\}}\\
&=& (1+c)\sum_{u^N}p_{U^N}(u^N) \text{Tr}\Big{\{}\Big{(}\sum_{u_1^{N-1}}\ketbra{u_1^{N-1}}{u_1^{N-1}}\otimes (I-\Lambda_{u_1^{N-1}u_N}^{B^N})\Big{)}\nonumber\\
&\Big{(}&\sum_{u_1^{N-1}}p_{U_1^{N-1}}(u_1^{N-1})\ketbra{u_1^{N-1}}{u_1^{N-1}}\otimes\overline{\mathcal{W}}^{u_1^N}(\rho^N)\Big{)}\Big{\}}\nonumber\\
&+& (2+c+c^{-1})\sum_{i=2}^{N-1}p_U(u_i)\text{Tr}\Big{\{}\Big{(}\sum_{u_1^{i-1}}\ketbra{u_1^{i-1}}{u_1^{i-1}}\otimes (I-\Lambda_{u_1^{i-1}u_i}^{B^N})\Big{)}\nonumber\\
&\Big{(}&\sum_{u_1^{i-1}}p_{U_1^{i-1}}(u_1^{i-1})\ketbra{u_1^{i-1}}{u_1^{i-1}}\otimes \overline{\mathcal{W}}^{u_1^i}(\rho^N)\Big{)}\Big{\}}\nonumber\\
&+& (2+c^{-1})\sum_{u_{1{}}}p_{U}(u_1)\text{Tr}\Big{\{}(I-\Lambda_{u_1}^{B^N})\overline{\mathcal{W}}^{u_1}(\rho^N)\Big{\}},
\label{Eq:82}
\end{eqnarray}where the first equality follows from Eq.~\ref{Eq:DefOverlineW2}. For the second equality, we have used the fact that \cite{Wilde2013}
\begin{equation}
\sum_x p_X(x)\text{Tr}\{A_x\rho_x\} = \text{Tr}\Big{\{}\Big{(}\sum_x \ketbra{x}{x}\otimes A_x\Big{)}\Big{(}\sum_{x'}p_X(x)\ketbra{x'}{x'}\otimes\rho_{x'}\Big{)}\Big{\}}.
\end{equation}Now, from the definition of synthesized channels present in Eq.~\ref{Eq:DefOverlineW} the next equality is produced
\begin{eqnarray}
\text{RHS of Eq.~\ref{Eq:82}}&=& (1+c)\sum_{u^N}p_{U^N}(u^N) \text{Tr}\Big{\{}\Big{(}\sum_{u_1^{N-1}}\ketbra{u_1^{N-1}}{u_1^{N-1}}\otimes (I-\Lambda_{u_1^{N-1}u_N}^{B^N})\Big{)}\mathcal{W}^{(N),u_N}(\rho^N)\Big{\}}\nonumber\\
&+& (2+c+c^{-1})\sum_{i=2}^{N-1}p_U(u_i)\text{Tr}\Big{\{}\Big{(}\sum_{u_1^{i-1}}\ketbra{u_1^{i-1}}{u_1^{i-1}}\otimes (I-\Lambda_{u_1^{i-1}u_i}^{B^N})\Big{)}\mathcal{W}^{(i),u_i}(\rho^N)\Big{\}}\nonumber\\
&+& (2+c^{-1})\sum_{u_{1}}p_{U}(u_1)\text{Tr}\Big{\{}(I-\Lambda_{u_1}^{B^N})\mathcal{W}^{(1),u_1}(\rho^N)\Big{\}}.
\label{Eq:84}
\end{eqnarray}Now, the number of terms in the previous sum can be reduced by means of a simple observation. 
Since the frozen bits are obtained by using a classical function, there is no need to make a measurement on the frozen bits position; i.e., 
$\Lambda_{u_1^{i-1}u_i}^{B^N} = I$ for any $i\in\mathcal{A}^c$. This leads to 
\begin{eqnarray}
\text{RHS of Eq.~\ref{Eq:84}}&<& \frac{(1+c+c^{-1})}{2}\sum_{i\in\mathcal{A}} Z_\rho(\mathcal{W}^{(i)}) = o(2^{-N^\beta}),
\end{eqnarray}where the definition of reliability for synthesized channels have been applied. The fact that this sum equals to $o(2^{-N^\beta})$
is a consequence of Proposition~\ref{Proposition:Preliminaries}.
\end{proof}

When one considers this decoding scheme for symmetric quantum reading, the same result as in Theorem~\ref{Theorem:ErrorProbDecay} is obtained. Thus, the strategy used is the same,
and there is only the need to properly design the pretty good measurement for the symmetric channels under consideration.

\section{Probe State Analysis}
\label{sec:inputDependenceAnalysis}
We are going to analyze how reliability and rate depend on the probe state used. The comparison will be between pure and entangled probe states. 
In both cases, a parameter of the probe states is maximized to give the optimal value of the quantity under consideration.

First of all, we adopt
\begin{equation}
\ket{\psi} = \sqrt{1-q}\ket{0} + e^{-i\phi}\sqrt{q}\ket{1}
\label{Eq:pure_probe}
\end{equation}as the pure probe state, where $q\in[0,1]$ and $\phi\in[0,2\pi)$. Because of the symmetric action of the AD channel 
with respect to the $z$-axis in the Bloch sphere representation, we can assume, without loss of generality, that $\phi = 0$. 
Secondly, the entangled state used as probe is given by
\begin{equation}
\ket{\Psi^{xz}} = (Z^zX^x\otimes I)\ket{\Phi},
\label{Eq:mixed_probe}
\end{equation}where $z,x\in\mathbb{Z}_2$, and $\ket{\Phi} = \sqrt{1-q}\ket{00} + \sqrt{q}\ket{11}$, with $q\in[0,1]$. Notice that for $q = 1/2$, we have 
all four Bell states. Using these pure and entangled states, we can compare their goodness for quantum reading assisted by 
polar codes. This comparison is made by optimizing the quantity under consideration with respect to $q$ and $q$.

Lets start considering pure probe states. Substituting Eq.~\ref{Eq:pure_probe} into Eq.~\ref{Eq:ADchannel}, it is possible to see that 
\begin{eqnarray}
	\mathcal{W}^{u_0}(\ketbra{\psi}{\psi}) &=& 
	\begin{pmatrix}
		1 - (1-u_0)q & \sqrt{(1-u_0)q(q-1)}\\
		\sqrt{(1-u_0)q(q-1)} & (1-u_0)q
	\end{pmatrix},\\
	\mathcal{W}^{u_1}(\ketbra{\psi}{\psi}) &=& 
	\begin{pmatrix}
		1 - (1-u_1)q & \sqrt{(1-u_1)q(q-1)}\\
		\sqrt{(1-u_1)q(q-1)} & (1-u_1)q
	\end{pmatrix}.
\end{eqnarray}Now, the minimal value of reliability with respect to $q$ can be computed. We are going to study a normalized version of reliability 
given by 
\begin{equation}
\overline{Z}(\mathcal{W})_{\ketbra{\psi}{\psi}} = \max_{q} \frac{Z(\mathcal{W})_{\ketbra{\psi}{\psi}}}{2\sqrt{p(p-1)}} = \max_{q} F(\mathcal{W}^{u_0}(\ketbra{\psi}{\psi}),\mathcal{W}^{u_1}(\ketbra{\psi}{\psi})).
\end{equation}For the current probe state, it is shown in Fig.~\ref{Fig:FidelityPureState} the minimal values of $\overline{Z}(\mathcal{W})_{\ketbra{\psi}{\psi}}$ and the corresponding $q$.
An important fact can be seen in Fig.~\ref{Fig:FidelityPureState}, the optimal pure probe state is given by the excited state $\ket{1}$. 
% This is an expected result but not for all values of $u_0$ and $u_1$.

\begin{figure}[h!]
	% \centering
	\hspace{-.5cm}
	\begin{subfigure}{.55\textwidth}
		\centering
		\includegraphics[width=\linewidth]{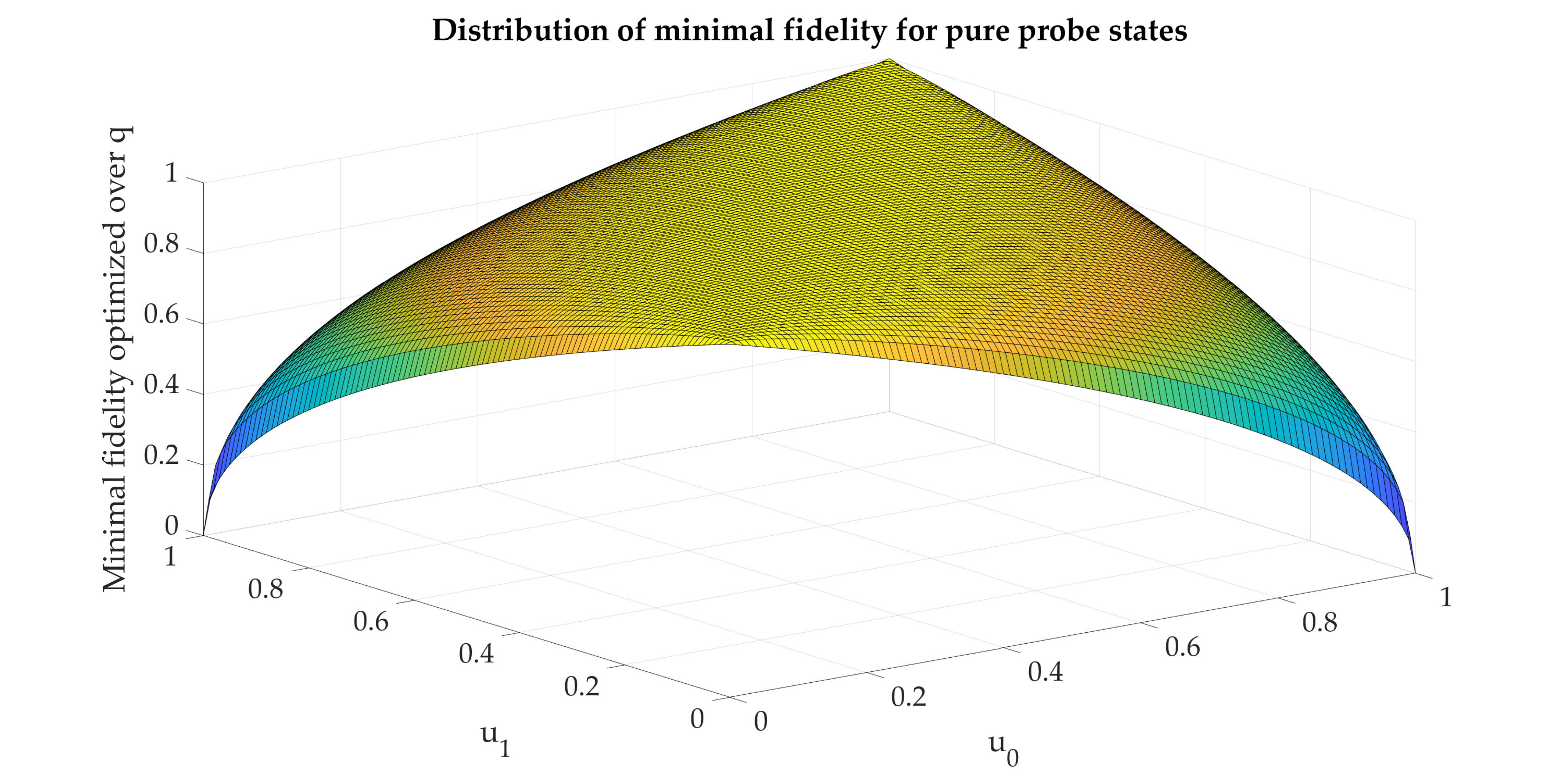}
		\caption{Minimal value of fidelity w.r.t. $q$.}
		\label{fig:sub1}
	\end{subfigure}%
	\begin{subfigure}{.55\textwidth}
		\centering
		\includegraphics[width=\linewidth]{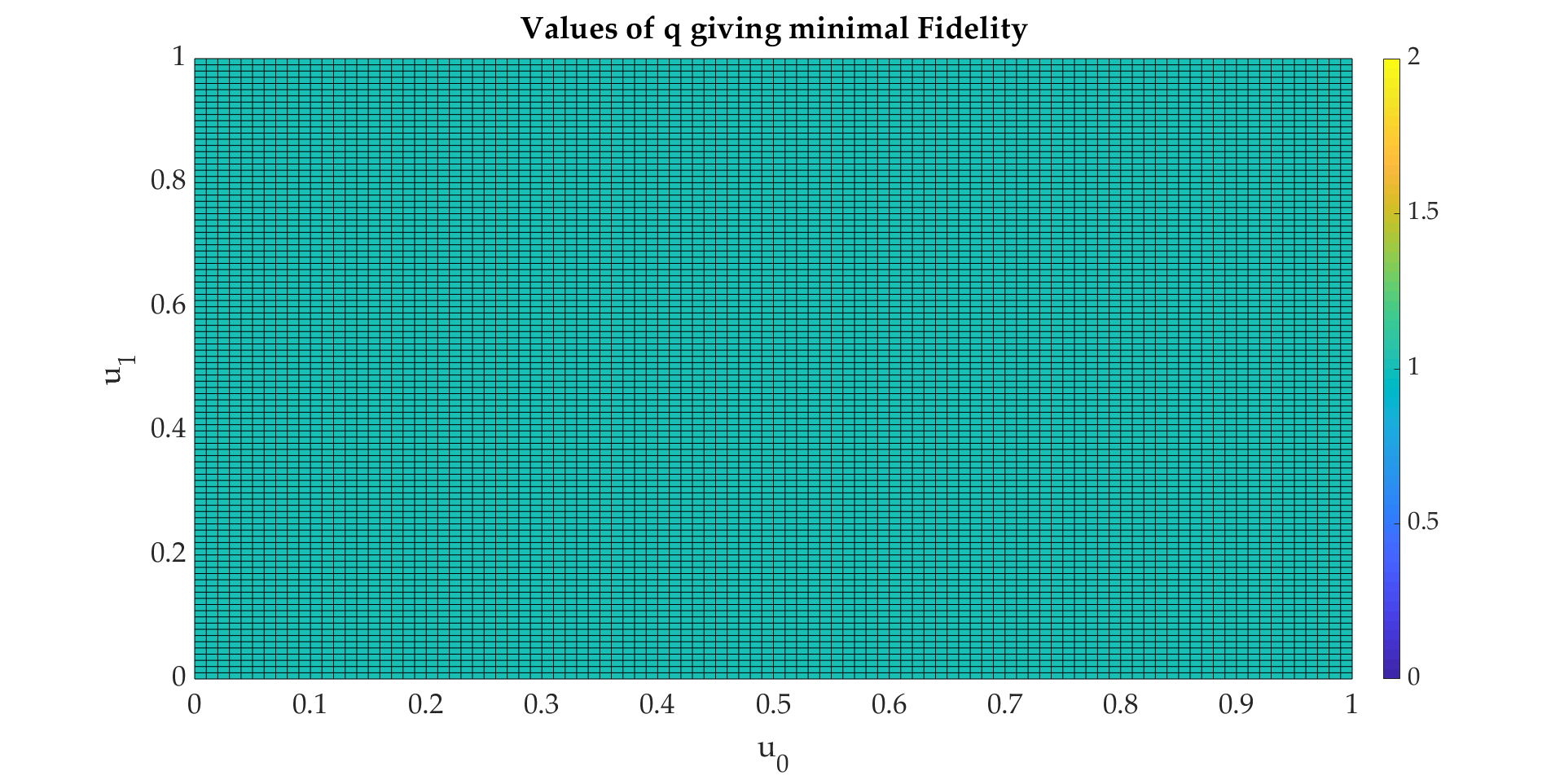}
		\caption{Distribution of $q$ giving the minimal value of fidelity.}
		\label{fig:sub2}
	\end{subfigure}
	\caption{Analysis of fidelity for the pure probe state.}
	\label{Fig:FidelityPureState}
\end{figure}

Considering mixed probe states, the same analysis can be done. There is a difference now on how to probe the channel. Since we have a bipartite entangled state, it can be used to probe the 
channel once or twice. The channel output in both of these situations is given by
\begin{eqnarray}
\textstyle
	\mathcal{W}^{u_i}\otimes\text{id}_{A_2}(\ketbra{\Psi^{xz}}{\Psi^{xz}}) &=& \Big{[} \Big{(} (1-q)\delta_{0x}\ketbra{0}{0}_{A_2} + q\delta_{1x}\ketbra{1}{1}_{A_2} \Big{)} \nonumber\\
	                                         &+& u_i\Big{(} (1-q)\delta_{1x}\ketbra{0}{0}_{A_2} + q\delta_{0x}\ketbra{1}{1}_{A_2}\Big{)}\Big{]}\ketbra{0}{0}_{A_1} \nonumber\\
	                                         &+& (\delta_{0z}-\delta_{1z})\sqrt{(1-u_i)q(1-q)}\Big{[}\delta_{0x}\ketbra{0}{1}_{A_2} + \delta_{1x}\ketbra{1}{0}_{A_2}\Big{]}\ketbra{0}{1}_{A_1}\nonumber\\
	                                         &+& (\delta_{0z}-\delta_{1z})\sqrt{(1-u_i)q(1-q)}\Big{[}\delta_{1x}\ketbra{0}{1}_{A_2} + \delta_{0x}\ketbra{1}{0}_{A_2}\Big{]}\ketbra{1}{0}_{A_1}\nonumber\\
	                                         &+& (1-u_i)\Big{[} (1-q)\delta_{1x}\ketbra{0}{0}_{A_2} + q\delta_{0x}\ketbra{1}{1}_{A_2} \Big{]}\ketbra{1}{1}_{A_1},
\label{Eq:HalfProbed}
\end{eqnarray}where $i=0,1$ and $\delta_{ab}$ is the Kronecker delta, and
\begin{eqnarray}
\textstyle
	\mathcal{W}^{u_0}\otimes\mathcal{W}^{u_1}(\ketbra{\Psi^{xz}}{\Psi^{xz}}) &=& [(1-q)\delta_{0x}+u_0(1-q)\delta_{1x}]\ketbra{0}{0}_{A_1}\otimes\ketbra{0}{0}_{A_2}\nonumber\\
											 &+& q(u_0\delta_{0x}+\delta_{1x})\ketbra{0}{0}_{A_1}\otimes[u_1\ketbra{0}{0}_{A_2} + (1-u_1)\ketbra{1}{1}_{A_2}]\nonumber\\
	                                         &+& (-1)^z\sqrt{(1-u_0)(1-u_1)q(1-q)}[\delta_{0x}\ketbra{0}{1}_{A_1}+\delta_{1x}\ketbra{1}{0}_{A_1}]\otimes\ketbra{0}{1}_{A_2}\nonumber\\
	                                         &+& (-1)^z\sqrt{(1-u_0)(1-u_1)q(1-q)}[\delta_{0x}\ketbra{1}{0}_{A_1}+\delta_{1x}\ketbra{0}{1}_{A_1}]\otimes\ketbra{1}{0}_{A_2}\nonumber\\
	                                         &+& (1-u_0)q\delta_{0x}\ketbra{1}{1}_{A_1}\otimes[u_1\ketbra{0}{0}_{A_2}+(1-u_1)\ketbra{1}{1}_{A_2}]\nonumber\\
	                                         &+& (1-u_0)(1-q)\delta_{1x}\ketbra{1}{1}_{A_1}\otimes\ketbra{0}{0}_{A_2}.
\label{Eq:FullProbed}
\end{eqnarray}For each value of $x,z$, we minimize the fidelity with respect to $q$ considering these two cases, which we call half-probed and full-probed cases. 
The first one corresponds to Eq.~\ref{Eq:HalfProbed} and the second to Eq.~\ref{Eq:FullProbed}. Our first description is over the half-probed case. See Figs.~\ref{Fig:HalfFidelityPsi10} and 
\ref{Fig:HalfFidelityPsi00} below. Notice that we have shown just the instance $z = 0$, since it gives the same results as for $z=1$. 

\begin{figure}[h!]
	% \centering
	\hspace{-.5cm}
	\begin{subfigure}{.55\textwidth}
		\centering
		\includegraphics[width=\linewidth]{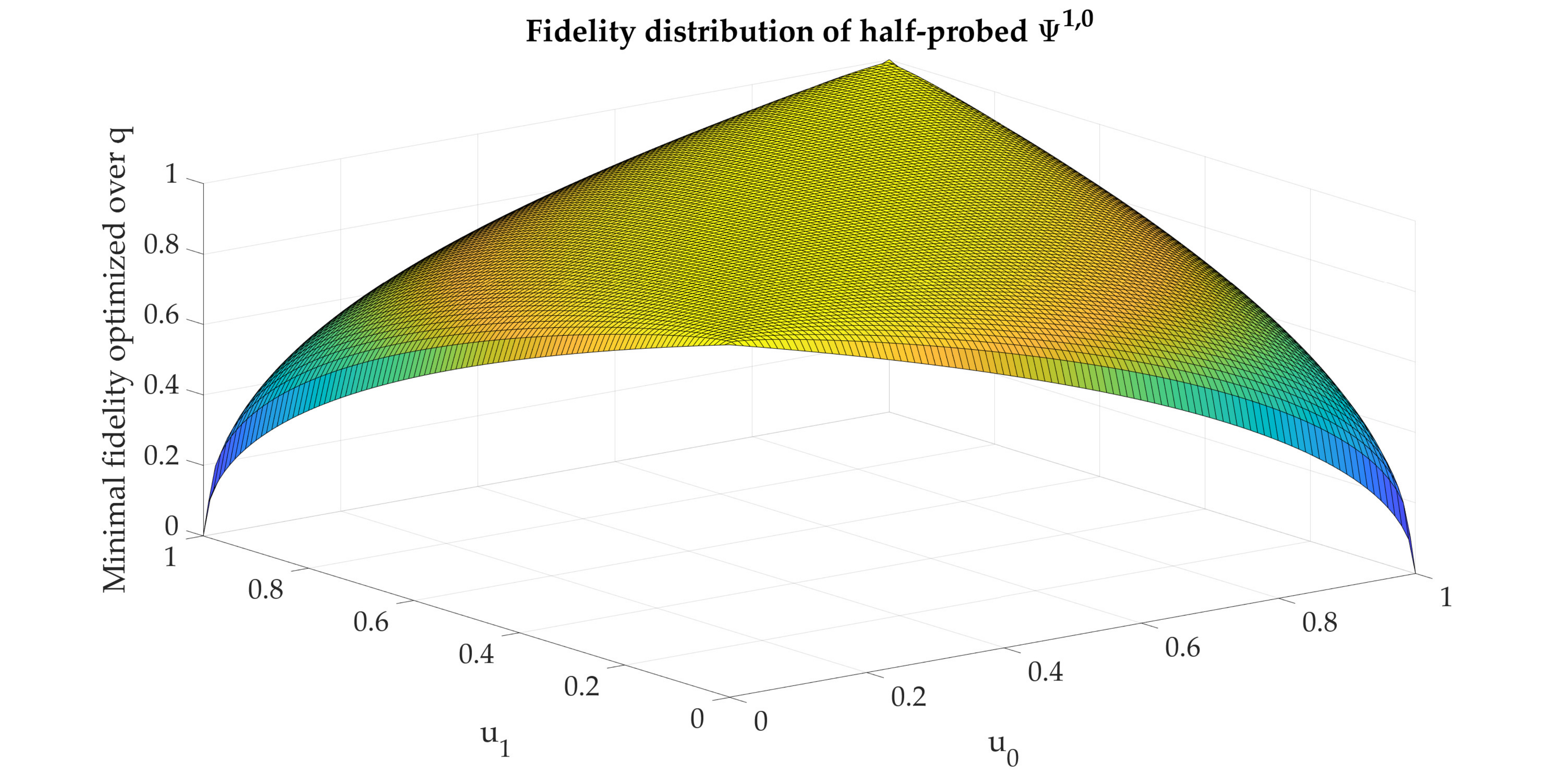}
		\caption{Minimal value of fidelity w.r.t. $q$.}
		\label{fig:sub1}
	\end{subfigure}%
	\begin{subfigure}{.55\textwidth}
		\centering
		\includegraphics[width=\linewidth]{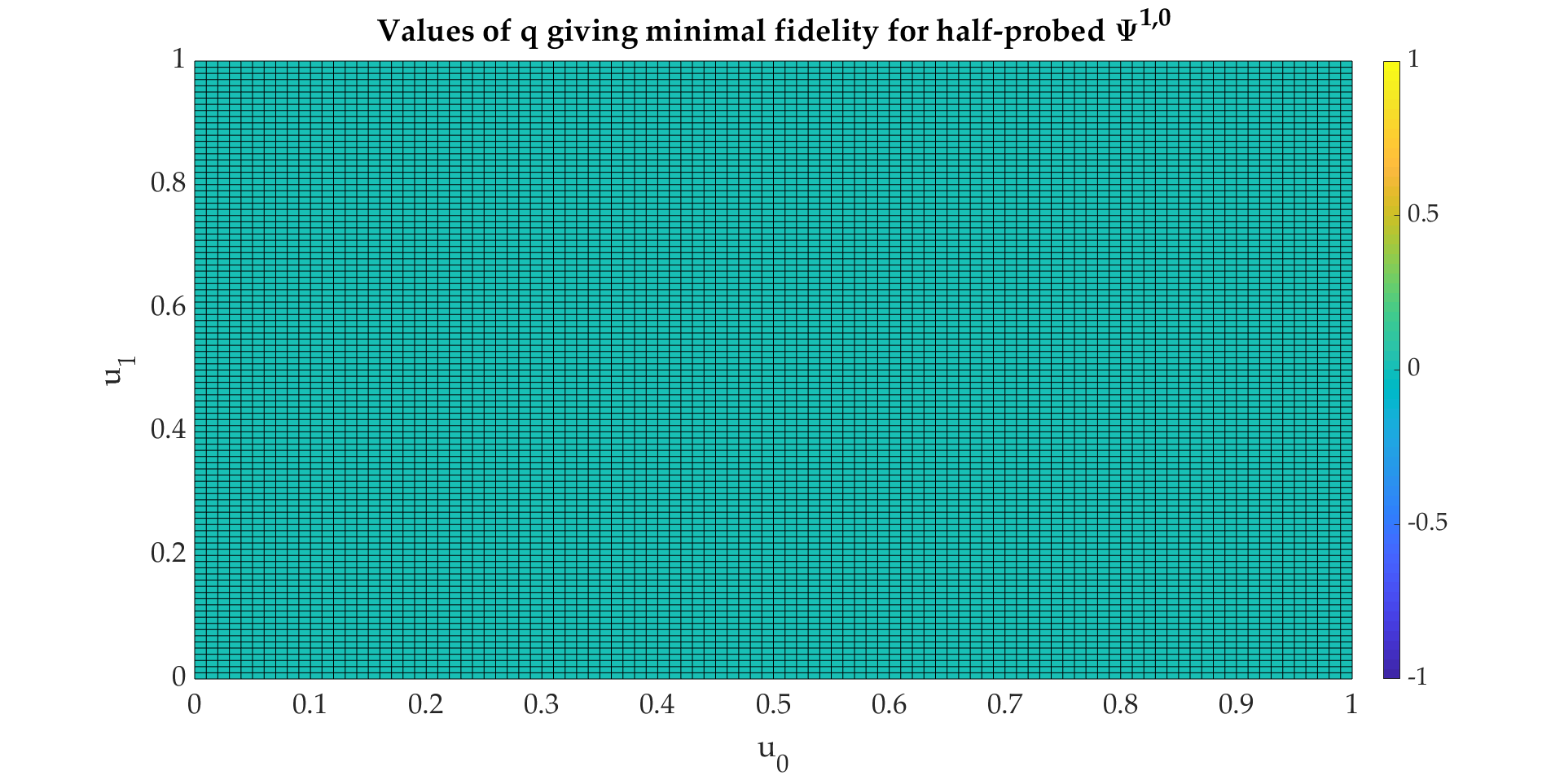}
		\caption{Distribution of $q$ giving the minimal value of fidelity.}
		\label{fig:sub2}
	\end{subfigure}
	\caption{Analysis of fidelity for the half-probed state $\ket{\Psi^{1,0}}$}
	\label{Fig:HalfFidelityPsi10}
\end{figure}

\begin{figure}[h!]
	% \centering
	\hspace{-.5cm}
	\begin{subfigure}{.55\textwidth}
		\centering
		\includegraphics[width=\linewidth]{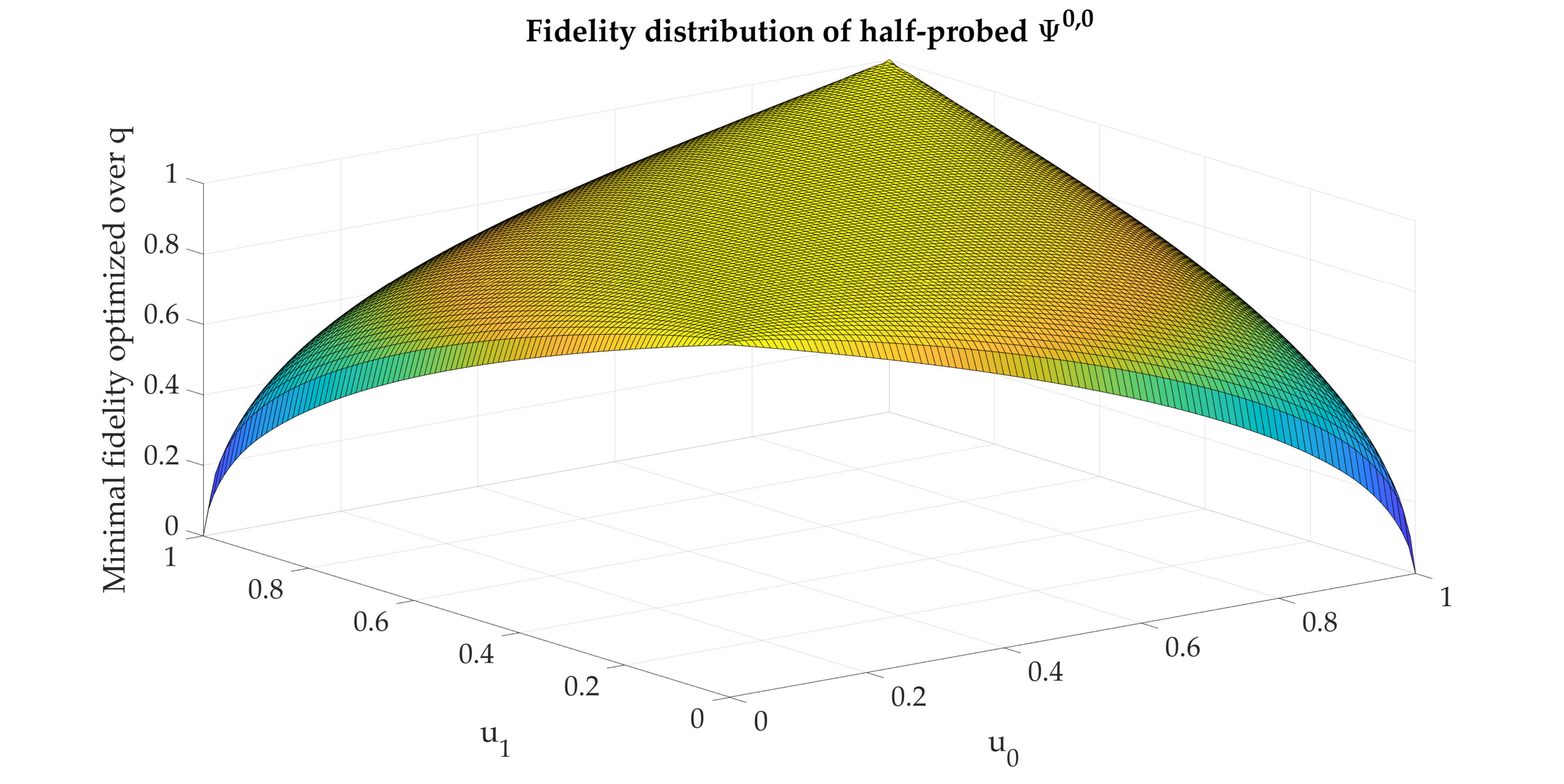}
		\caption{Minimal value of fidelity w.r.t. $q$.}
		\label{fig:sub1}
	\end{subfigure}%
	\begin{subfigure}{.55\textwidth}
		\centering
		\includegraphics[width=\linewidth]{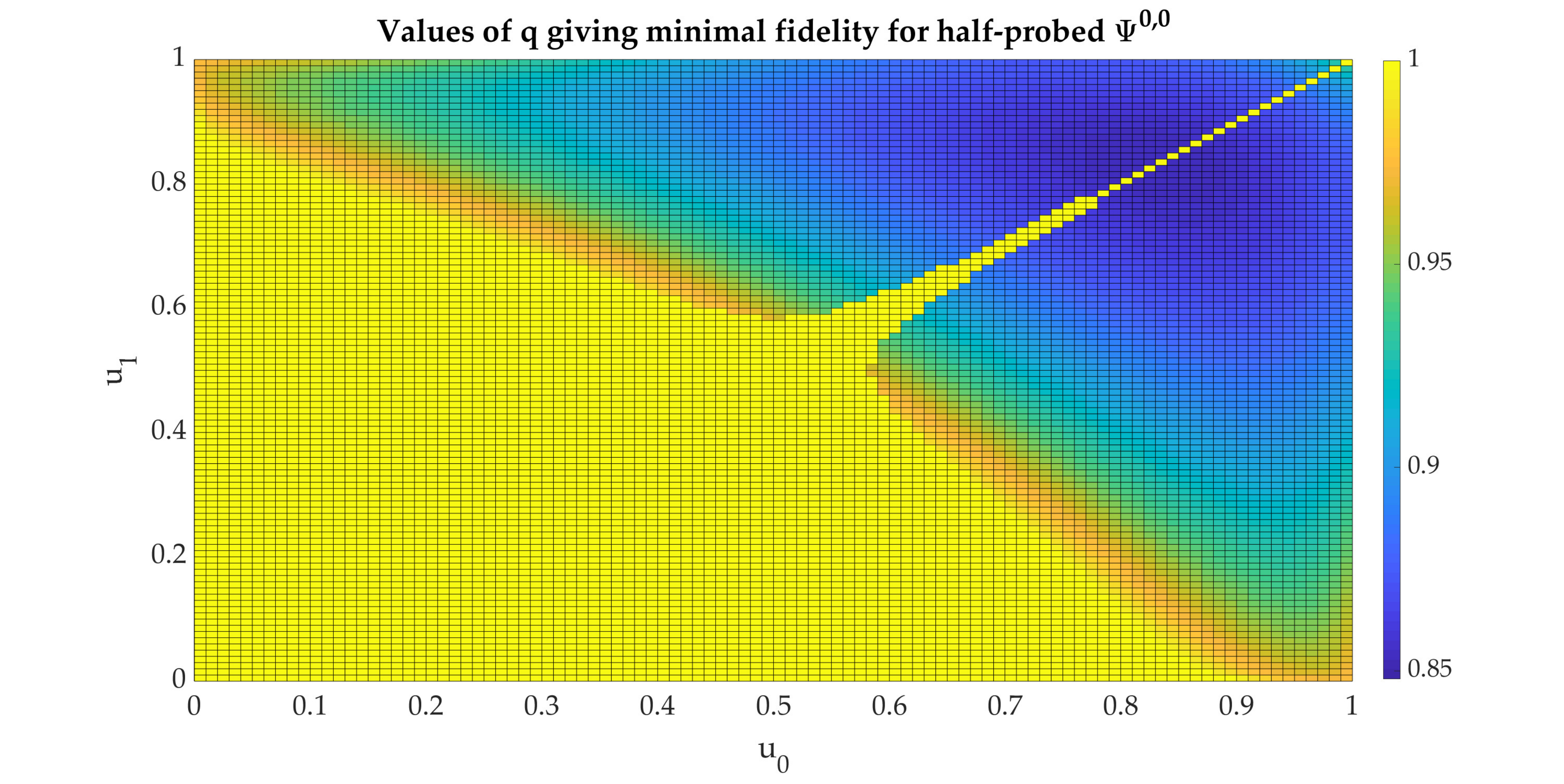}
		\caption{Distribution of $q$ giving the minimal value of fidelity.}
		\label{fig:sub2}
	\end{subfigure}
	\caption{Analysis of fidelity for the half-probed state $\ket{\Psi^{0,0}}$.}
	\label{Fig:HalfFidelityPsi00}
\end{figure}

\begin{figure}
		\centering
		\includegraphics[width=0.8\linewidth]{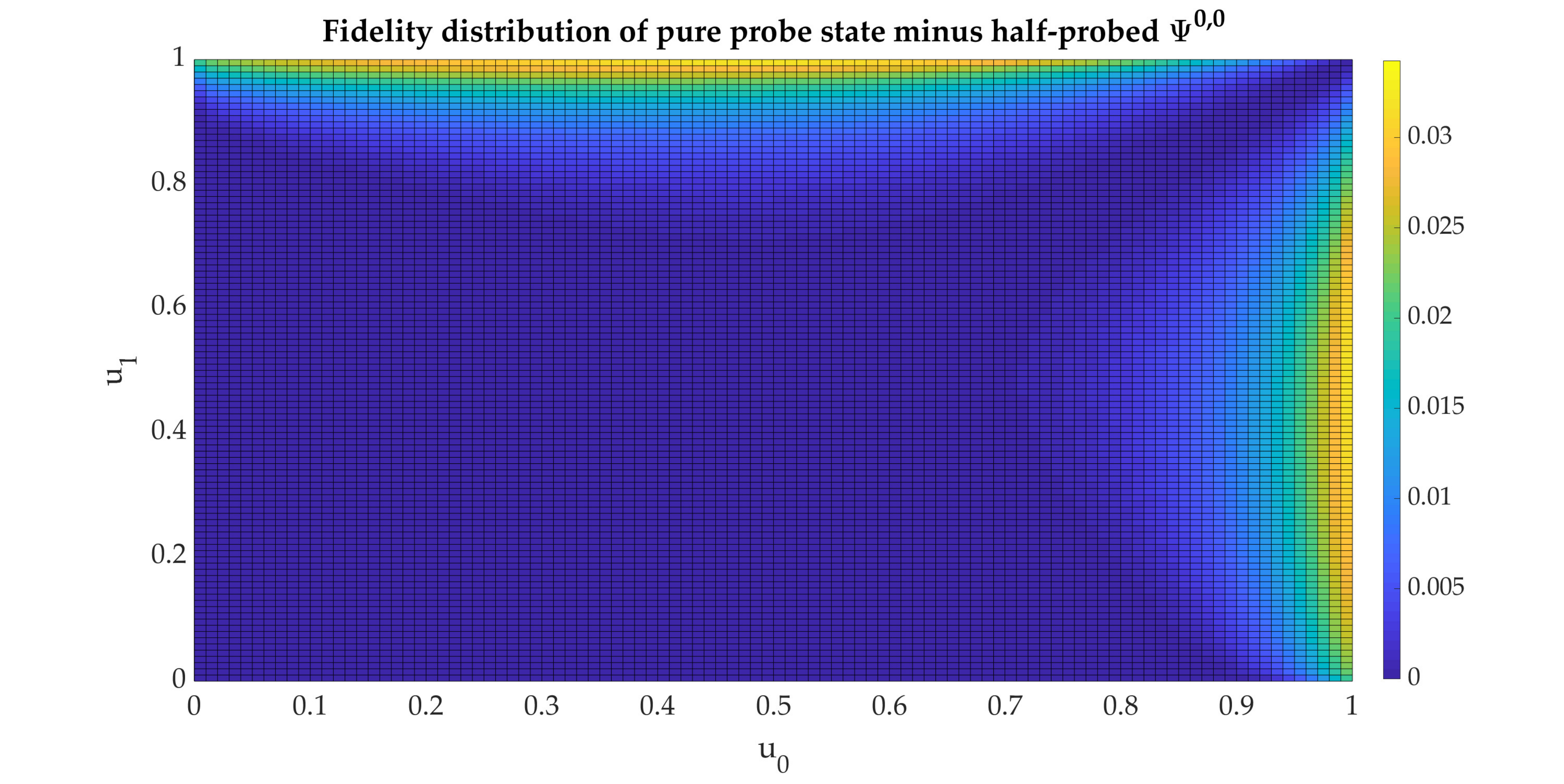}
		\caption{Fidelity of pure state minus fidelity of half-probed state $\ket{\Psi^{0,0}}$.}
		\label{Fig:FidelityPure-HalfFidelityPsi00}
\end{figure}

Two conclusions can be obtained from Figs.~\ref{Fig:HalfFidelityPsi10} and \ref{Fig:HalfFidelityPsi00}. Firstly, the shape of the fidelity of half-probed state $\ket{\Psi^{10}}$ is exactly the 
same as for the pure probe state case. This is clear from the fact that the value of $q$ giving maximal fidelity is $q = 0$, which turns the probe state to be $\ket{1}$. 
As the second point, we see in Figs.~\ref{Fig:HalfFidelityPsi10} and \ref{Fig:HalfFidelityPsi00} that half-probing with the state $\ket{\Psi^{00}}$ can give some improvement in 
the channel discrimination task. This improvement is obtained where $q\neq 1$. However, the difference is not so significant, making the use of $\ket{\Psi^{00}}$ justified only if 
entanglement cost is negligible.

The subsequent analysis is over a full-probed state. In opposition to what has been shown in the half-probed case, there is no improvement in using full-probed $\ket{\Psi^{00}}$. 
Fig.~\ref{Fig:fullFidelityPsi00} shows that $q=1$ gives the minimal value of fidelity in this case. This leads to the fidelity having values equal to the square of the 
pure probe state. Thus, optimizing full-probed strategy using $\ket{\Psi^{00}}$ leads to the same strategy for pure probe state. Lastly, we consider full-probed strategy 
using $\ket{\Psi^{10}}$. We have $q=0$ giving the optimal value of fidelity. A consequence of this is that now we probe the quantum channel twice using the state $\ket{0}$ and 
$\ket{1}$ in each round. Since for pure probe, we are using the state $\ket{1}$, then the round probing with $\ket{0}$ contribute degrading the fidelity. Therefore, the optimal fidelity 
obtained using full-probed $\ket{\Psi^{10}}$ is much higher than using full-probed $\ket{\Psi^{00}}$,
as can be seen in Fig.~\ref{Fig:fullFidelityPsi10-fullFidelityPsi00}.

\begin{figure}[h!]
	% \centering
	\hspace{-.5cm}
	\begin{subfigure}{.55\textwidth}
		\centering
		\includegraphics[width=\linewidth]{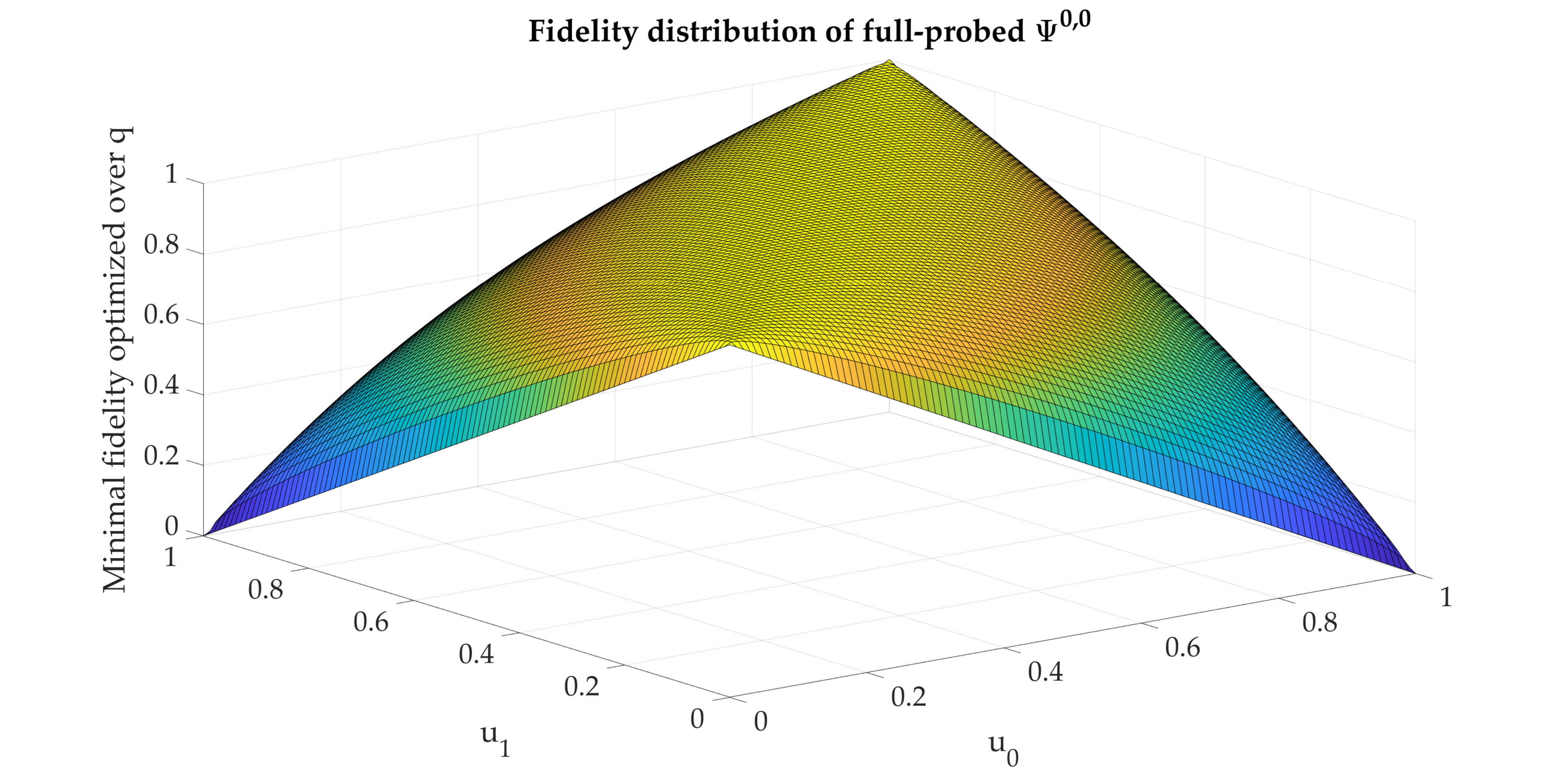}
		\caption{Minimal value of fidelity w.r.t. $q$.}
		\label{fig:sub1}
	\end{subfigure}%
	\begin{subfigure}{.55\textwidth}
		\centering
		\includegraphics[width=\linewidth]{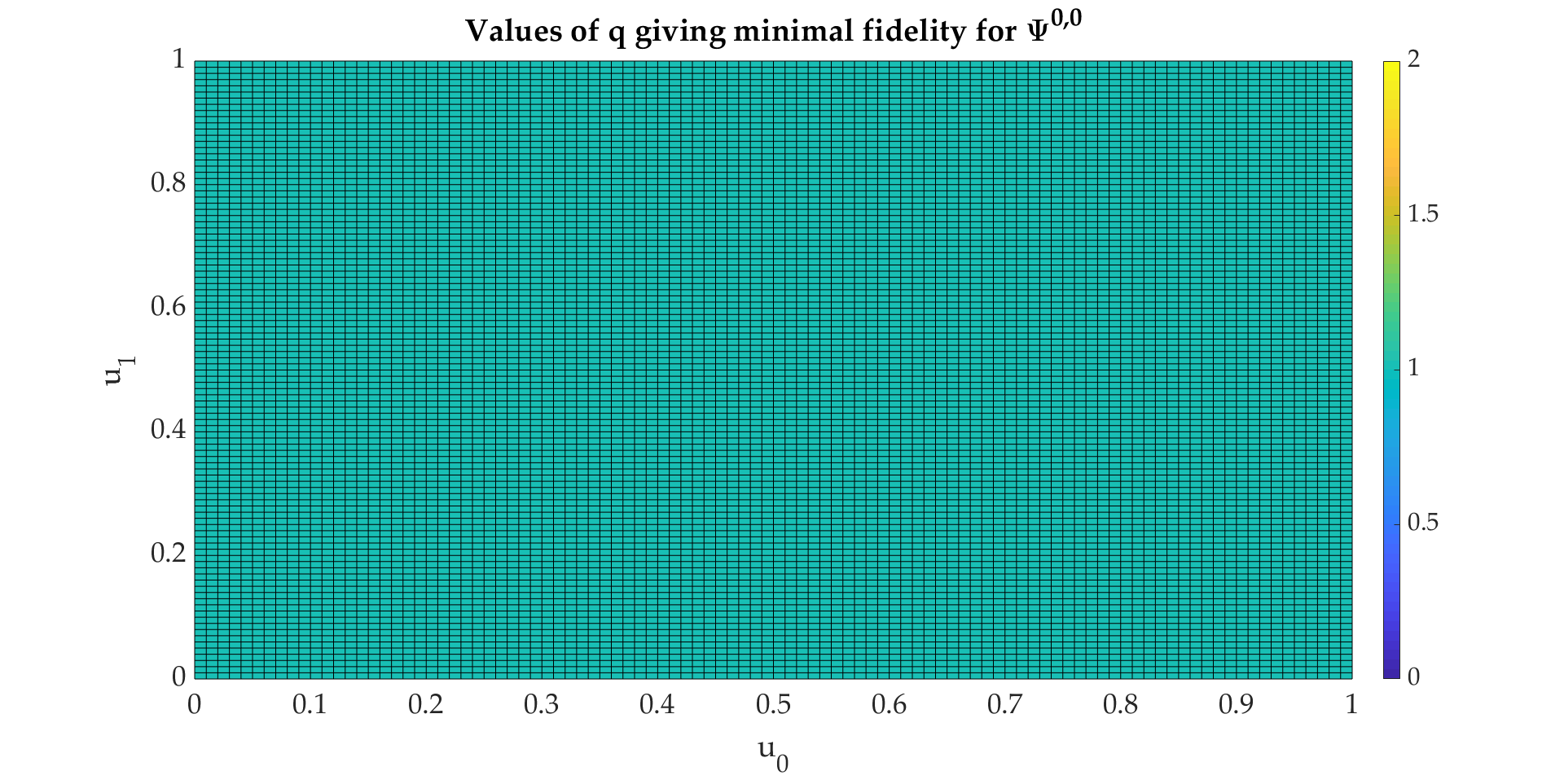}
		\caption{Distribution of $q$ giving the minimal value of fidelity.}
		\label{fig:sub2}
	\end{subfigure}
	\caption{Analysis of fidelity for the full-probed state $\ket{\Psi^{0,0}}$.}
	\label{Fig:fullFidelityPsi00}
\end{figure}

\begin{figure}[h!]
	% \centering
	\hspace{-.5cm}
	\begin{subfigure}{.55\textwidth}
		\centering
		\includegraphics[width=\linewidth]{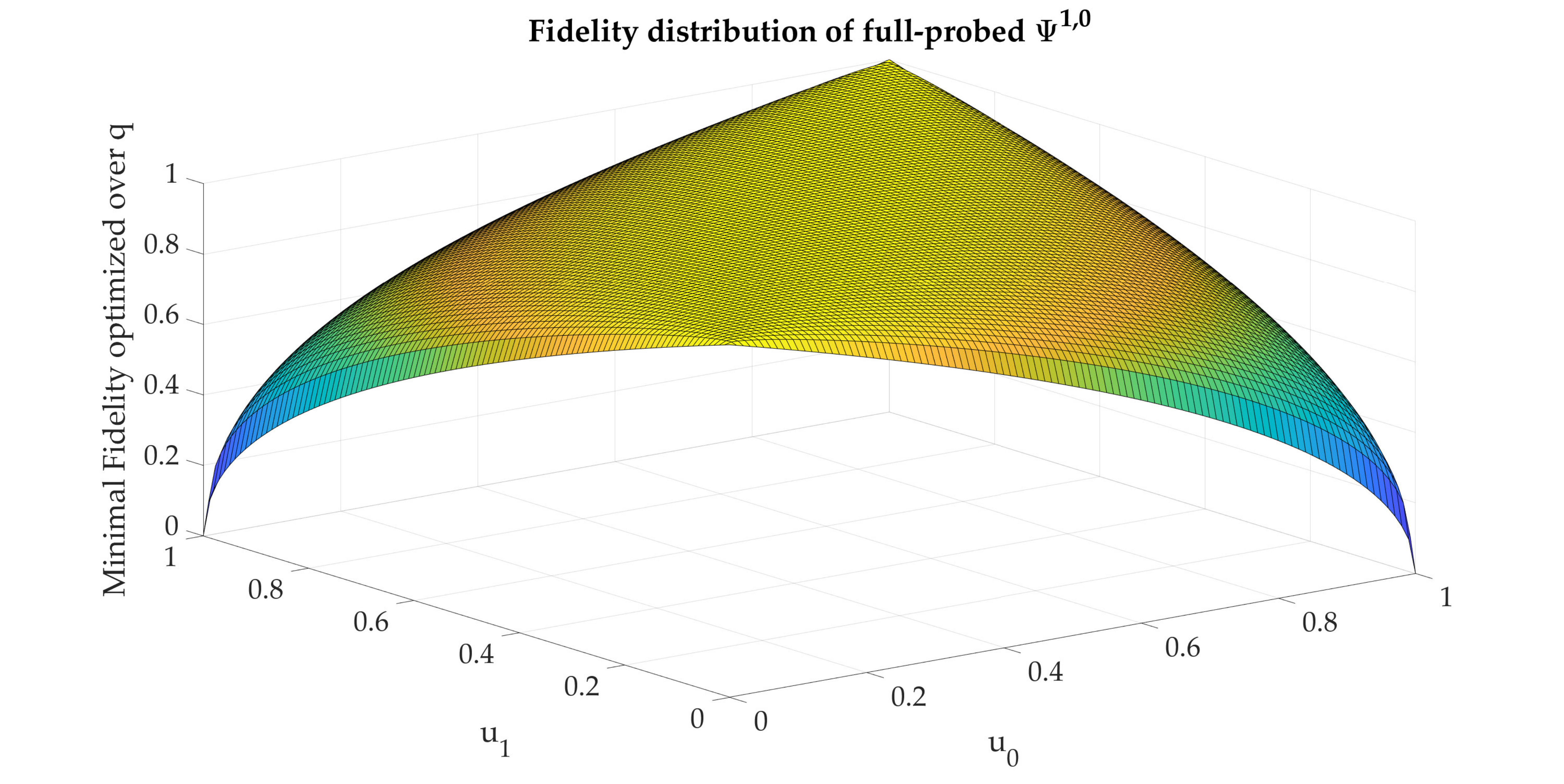}
		\caption{Minimal value of fidelity w.r.t. $q$.}
		\label{fig:sub1}
	\end{subfigure}%
	\begin{subfigure}{.55\textwidth}
		\centering
		\includegraphics[width=\linewidth]{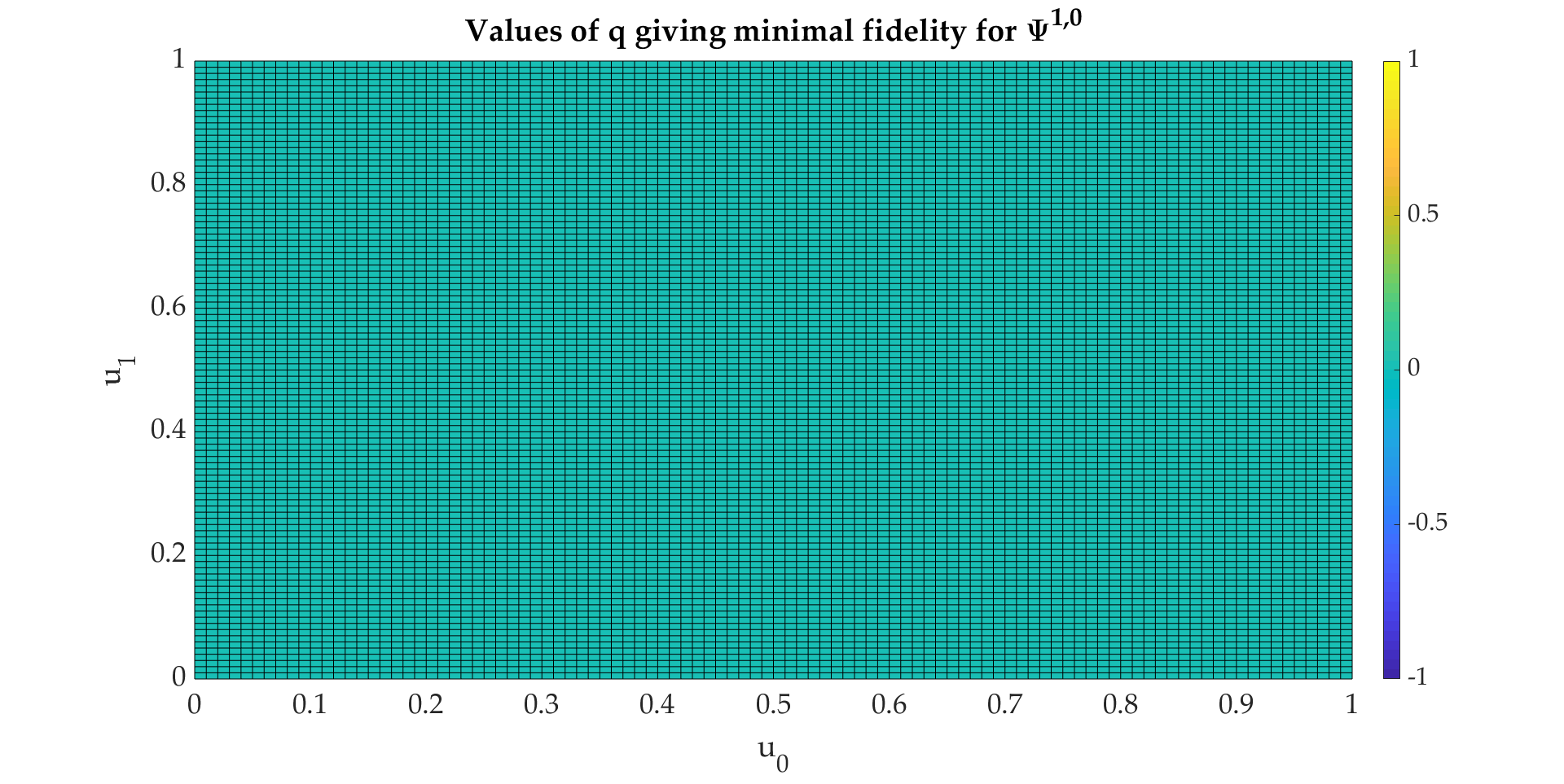}
		\caption{Distribution of $q$ giving the minimal value of fidelity.}
		\label{fig:sub2}
	\end{subfigure}
	\caption{Analysis of fidelity for the full-probed state $\ket{\Psi^{1,0}}$}
	\label{Fig:FidelityPsi10}
\end{figure}

\begin{figure}[h!]
		\centering
		\includegraphics[width=0.8\linewidth]{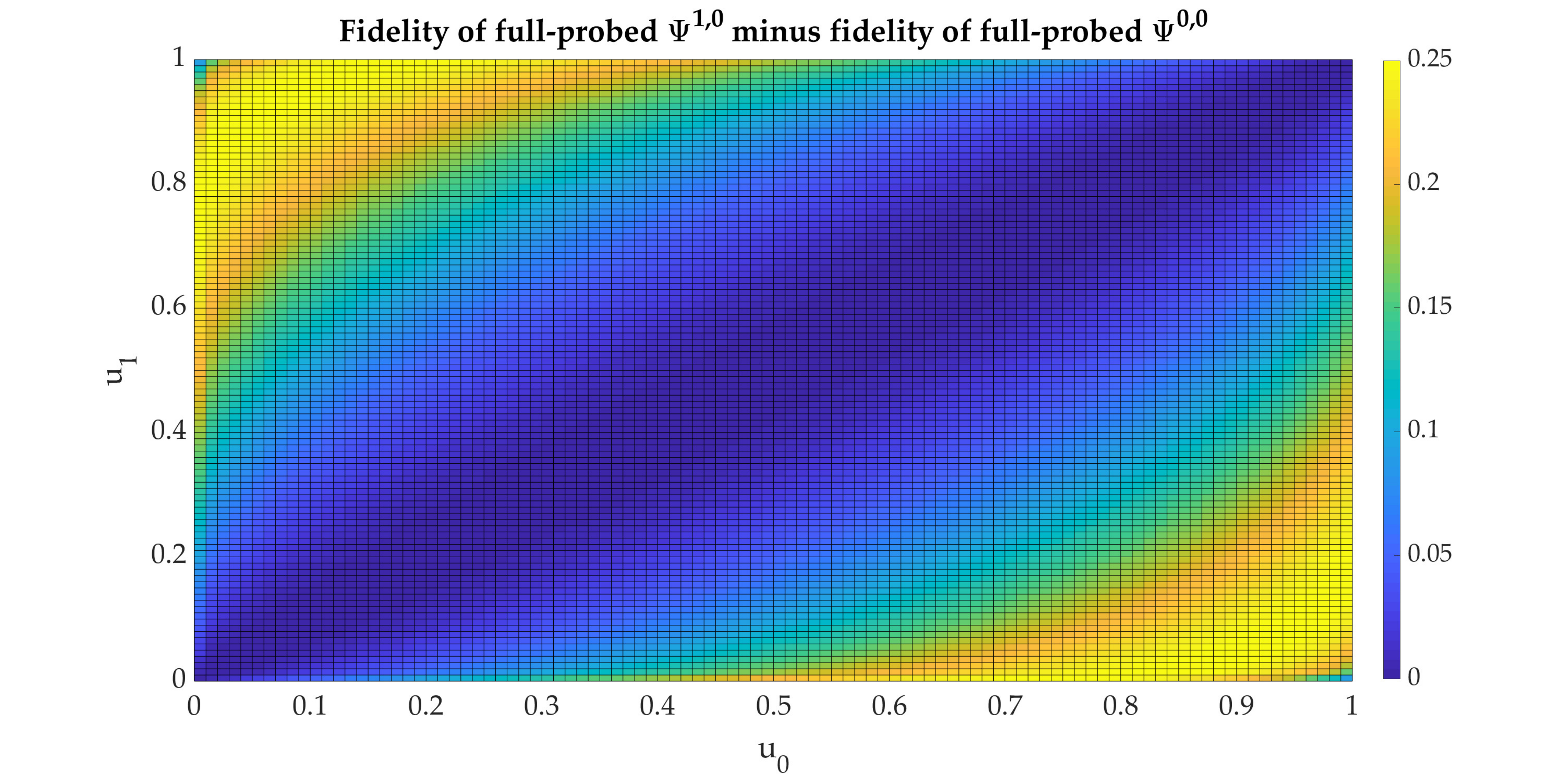}
		\caption{Fidelity of full-probed state $\ket{\Psi^{1,0}}$ minus fidelity of full-probed state $\ket{\Psi^{0,0}}$.}
		\label{Fig:fullFidelityPsi10-fullFidelityPsi00}
\end{figure}

The previous results can be extended to the mutual information between the classical system $X$ and the quantum channel output system $B$. 
In Proposition~\ref{Prop:MutualInformationAndReliability} we have seen that rate and reliability are connected in a way that when the rate 
is close to its maximum, then the reliability is close to its minimum, and vice versa. Thus, we can conclude the following. For channels with 
values of $u_0$ and $u_1$ close, the mutual information will be close to zero. On the other hand, when one of the $u_i$, $i=0,1$, is close to unit 
and the other is close to zero, the mutual information attains its maximum value possible. Furthermore, probing the quantum memory cell using 
the half-probed $\ket{\Psi^{00}}$ strategy is again the optimal procedure for quantum reading.

\section{Final Remarks}
\label{sec:Conclusion}
We have demonstrated a new polar coding scheme for quantum memory cell discrimination. 
To achieve this goal, we had to introduce new definitions on the rate and reliability of a 
quantum channel. The polarization phenomenon produced by channel combining and analyzed in 
the channel splitting part has been shown for these two quantities. This was established due to an 
inequality connecting both of them. In the channel splitting part, we have also introduced the synthesized 
channels created by the polar coding. It had been shown that when the number of channels is arbitrarily large, 
the set of synthesized channels can be divided in two groups, good and bad channels. Additionally, the fraction 
of such channels is related to the mutual information of the original quantum channel in consideration. This result 
has motivated us to construct an optimal encoding scheme for quantum reading. A decoding scheme 
was introduced and analyzed, as well. Using an existence proof of pretty good measurements given in this paper and a 
previous quantum union bound from the literature, we have shown that our decoding scheme has error probability that 
decays exponentially fast with respect to the code length. At the end, optimizations over probe states are investigated, 
leading to the conclusion of half-probed states as the best choice, in general, of probe states to be used. 

This paper has given some future investigation topics. A question that can be stated is how to apply polar 
coding to more general classical and quantum systems. First of all, we can extend the binary discrimination to a $d$-ary discrimination 
task. The mathematical equivalence of this is considering a classical system with a larger alphabet. Some classical polar codes 
have been proposed in the literature and, with the proper adjustments, applied here. Secondly, we can consider the set of 
quantum memory cells to be composed of generalized amplitude damping channel. This class of channel can be seen as a second-order 
approximation of classical digital memory, where the model takes the environment temperature into consideration. In this open question, 
the task would be to find the optimal probe state for channel discrimination. Third, and more important, how can one construct efficient 
polar codes for Gaussian bosonic channel discrimination. This still a research topic even for classical Gaussian channel. 
Gaussian bosonic channel is the ultimate goal as quantum channel model for classical digital memories. Because this class of 
channels is defined over continuous variables, attacking it needs to be two-fold. The polar code scheme has to take into account 
the channel to prove the polarization phenomenon. Additionally, an optimization over the probe states is necessary under 
energy constrain. These two approaches are not independent, which makes the task even more difficult. 
Lastly, we do not know if there exists a provable optimal probe state in the contexts mentioned. Uniqueness also needs to be verified.

% Two important open questions are how to properly apply polar codes to Gaussian channels and, 
% in particular, to Gaussian bosonic channels. As the ultimate model of memory cell is given by 
% lossy bosonic channels, this is also an important open question to quantum memory cell discrimination. 
% A possible approach is to gradually construct attractive polar codes to finite dimensional approximations of 
% lossy bosonic channels. The amplitude damping channel adopt in this paper is one of these approximations. 
% Constructing polar codes to generalized amplitude damping channels is a interesting path to follow. Additionally, 
% the existence and uniqueness of optimal probe states still an open question for discrimination an arbitrarily large 
% set of lossy bosonic channels. Thus, there exists several fronts to approach the problematic of lossy bosonic channels
% discrimination.

% puts also forward issues for future investigations.
% I WOULD THEN MENTION THE FOLLOWING:
% 1) extension from binary to d-ary discrimination;
% 2) extension to generalized amplitude damping channel (a model for qubit that include temperature)
% 3) extension to Gaussian bosonic channels (use qubit and qudit models as successive approximations)
% 4) issue of existence and uniqueness of optimal probe states in these new contexts.

%The End 

\section{Acknowledgments}
The authors acknowledge the financial support of the Future and Emerging Technologies (FET)
programme, within the Horizon-2020 Programme of the European Commission, under the FET-Open 
grant agreement QUARTET, number 862644.

% ---- Bibliography ----
\addcontentsline{toc}{section}{References}
\nocite{*}
\printbibliography

\end{document}